\documentclass[conference]{IEEEtran}
\IEEEoverridecommandlockouts
\usepackage{amsmath}
\usepackage{amssymb}
\usepackage{mathtools,amsthm}
\usepackage{comment}
\usepackage{graphicx}
\usepackage{caption}
\usepackage{subcaption}
\usepackage{hyperref}
\usepackage{xcolor}
\newtheorem{theorem}{Theorem}

\theoremstyle{definition}
\newtheorem{definition}{Definition}
\newtheorem{proposition}{Proposition}
\usepackage[ruled,linesnumbered]{algorithm2e}
\usepackage{enumitem}
\def\BibTeX{{\rm B\kern-.05em{\sc i\kern-.025em b}\kern-.08em
    T\kern-.1667em\lower.7ex\hbox{E}\kern-.125emX}}
\begin{document}

\title{Matching-with-Contracts for the AI-RAN Market: AIGC-as-a-Service for Teleoperation}

\author{Zijun~Zhan,
        Yaxian~Dong,
        Daniel~Mawunyo~Doe,
        Yuqing~Hu,
        Shaohua Cao,
        and~Zhu~Han,%  ~\IEEEmembership{Fellow,~IEEE} <-this % stops a space
        \thanks{Zijun Zhan is with the Department of Electrical and Computer Engineering, University of Houston, 4800 Calhoun Rd, Houston, TX 77004, USA. E-mail: zzhan@uh.edu} 
         \thanks{Yaxian Dong and Yuqing Hu are with the Department of Architectural Engineering, The Pennsylvania State University, University Park, PA 16802, USA (E-mail: yzd5221@psu.edu and yfh5204@psu.edu)}
	\thanks{Daniel Mawunyo Doe is with the Department of Electrical and Computer Engineering, Prairie View A\&M University, 100 University Dr, Prairie View, TX 77446, USA. Email: dmdoe@pvamu.edu}
         \thanks{Shaohua Cao is with the Qingdao Institute of Software, College of Computer Science and Technology, China University of Petroleum (East China), Qingdao 266580, China. E-mail:shaohuacao@upc.edu.cn}
         \thanks{Zhu Han is with the Department of Electrical and Computer Engineering, University of Houston, 4800 Calhoun Rd, Houston, TX 77004, USA. E-mail: hanzhu22@gmail.com}
}

\maketitle

\begin{abstract}
Artificial intelligence radio access networks (AI-RANs) are a promising architecture for bolstering the prosperity of the edge AI ecosystem. A well-designed incentive mechanism can further ensure the sustainable development of this ecosystem. However, incentive mechanism design faces two major challenges: 1) information asymmetry, where AI-RAN operators have only partial knowledge of AI users’ utility functions, and 2) competition, as multiple AI-RAN operators coexist in real-world markets. Remarkably, chaotic and adversarial competition might compromise AI-RAN operators' utility. To this end, we develop a matching-with-contracts framework for incentive mechanism design in AI-RAN service markets. The framework extends the static matching-with-contracts model by jointly characterizing the contract design of multiple competitive operators, user-operator matching, and dynamic evolution of the market state. Specifically, the incentive mechanism offered by each AI-RAN operator takes the form of a contract menu, where each contract item consists of an AI service latency agreement and a corresponding price. We model the AI service process as three independent queues and characterize the violation probability of the latency agreement using queueing theory and the Chernoff bound. To derive an effective incentive mechanism, we further propose a mixed stable matching-with-contracts algorithm that jointly updates user-side matching decisions and operator-side contract menus. Simulation results for a teleoperation-oriented AIGC service demonstrate the effectiveness and robustness of the proposed method. Compared with benchmark schemes, our method improves the total utility of AI-RAN operators by at least 56.8\% under representative settings.
\end{abstract}

\begin{IEEEkeywords}
AI-RAN, information asymmetry, incentive mechanism design, task offloading, edge AI
\end{IEEEkeywords}

\section{Introduction} \label{sec:1}
\IEEEPARstart{T}{he} rapid convergence of artificial intelligence and wireless communication is giving rise to the paradigm of artificial intelligence radio access network (AI-RAN), which aims to tightly integrate AI capabilities with radio access network operations and infrastructures \cite{khan2023airan, bonati2021intelligence, giannopoulos2022supporting}. In general, AI-RAN can be classified into three complementary categories, AI-for-RAN, AI-on-RAN, and AI-and-RAN \cite{kundu2025ai}. AI-for-RAN uses AI techniques to optimize conventional RAN functions such as scheduling, resource allocation, mobility management, and interference mitigation. AI-on-RAN treats the RAN as a service platform that can provide distributed communication and computing resources for AI task execution, thereby enabling AI services to be deployed and delivered over network infrastructures. AI-and-RAN further emphasizes the deep co-design of AI services and RAN systems, where communication and intelligence are jointly optimized in a unified architecture. Among these three directions, AI-RAN is increasingly viewed as a key enabler of future wireless networks because it not only enhances network intelligence but also opens a new path for turning the RAN into an active infrastructure for AI service provisioning.

Existing studies have extensively explored the AI-for-RAN direction, where AI methods have been introduced to improve network control and operational efficiency \cite{zhou2021ranslicing, kouchaki2022actor, rezazadeh2023madrloran, lai2023loadbalancingoran, mahmoud2024datadrivenoran, rezazadeh2024steporan}. By comparison, research on AI-on-RAN is still at an infancy stage, although it is attracting increasing attention as AI services become more computation-intensive, latency-sensitive, and geographically distributed \cite{khan2023airan, bonati2021intelligence, giannopoulos2022supporting}. Current work in this area has mainly focused on system architecture design, communication-computation resource orchestration, service placement, and performance optimization under network-side control \cite{bonati2021intelligence, giannopoulos2022supporting, filali2023commcomporan, abedin2022elastic}. While these efforts provide important technical foundations, they often abstract away the strategic behavior of participants and implicitly assume that the required resources can be coordinated without sufficient economic incentives. In practice, however, AI-on-RAN involves multiple self-interested entities with heterogeneous capabilities, costs, and service demands. Without an effective incentive mechanism, AI-RAN operators may lack motivation to participate, and AI users may not be efficiently matched with suitable RAN resources. This gap highlights the necessity of incentive mechanism design for the AI-on-RAN paradigm.

Designing such an incentive mechanism is challenging for two main reasons. First, competition is inherent in the AI-on-RAN since the market in the real-world is non-monopoly. Specifically, different AI-RAN operators may compete to attract profitable tasks under their own capacity and quality constraints. The competition makes the incentive mechanism design for AI-RAN operators highly interdependent. Second, information asymmetry is unavoidable. AI users typically possess private information about their valuations, costs, and opportunity losses, which is hidden from AI-RAN operators. If the information asymmetry is not properly handled in the incentive mechanism, strategic misreporting may distort matching outcomes and reduce overall market efficiency. Therefore, the incentive mechanism design problem in AI-on-RAN must simultaneously address competitive interactions and asymmetric information. We summarize the research question in this paper as: How can we design an effective incentive mechanism for AI-on-RAN that \textit{maximizes the sum of AI-RAN operators' utility under competition and information asymmetry?}

A natural way to tackle the aforementioned research is to leverage economic and game-theoretic tools that have been used to study resource trading and incentive design in wireless networks, edge computing, and digital platforms \cite{hatfield2005matching, zhou2019vfccontractmatching, su2020hierarchicalmec, su2021mecmatchingcontracts, lim2021towards, hui2022quality, wang2024cvflcontract, yao2025saflcontract, ye2025aigccontract, zhou2025aoicontract, picano2025llmmatching}. In particular, contract theory is well-suited to addressing information asymmetry, since it has the self-revealing property via differentiated contract items \cite{li2023book}. Matching theory, on the other hand, is effective for modeling decentralized competition and bilateral preferences in multi-agent markets\cite{lovasz2009matching}. Among related approaches, the matching-with-contracts framework is especially appealing because it integrates agent matching and contract selection into a unified market structure \cite{hatfield2005matching}. Notably, existing matching-with-contracts models cannot be directly applied to our problem since they consider a static setting and the contract menu is given ex-ante. However, the AI-on-RAN market is shaped by dynamic interactions between AI users and AI-RAN operators.

To tackle the research question posed in this paper, we develop a unified framework that combines contract design with competitive matching for AI-on-RAN service provisioning. Specifically, we consider a dynamic market in which heterogeneous AI users offload AI inference tasks to multiple competing AI-RAN operators through latency-price contract menus. We model the AI service process as a three-stage queue, which includes uplink transmission, AI inference processing, and downlink transmission. We characterize the latency agreement violation probability and the utilities of both AI users and AI-RAN operators. Subsequently, we formulate the contract design problem under asymmetric information and competition. Notably, user selection and contract menus designed by AI-RAN operators jointly shape the market outcome, which hinders the derivation of the incentive mechanism. Therefore, we propose a mixed stable matching with contracts algorithm to iteratively update contract menus and matching decisions.

In a nutshell, the main contributions of this paper are as follows.
\begin{enumerate}
    \item We propose a incentive design framework for competitive AI-on-RAN service provisioning under information asymmetry. Specifically, the proposed framework jointly captures latency-price contract design, user-side operator matching, and congestion-dependent AI service provisioning. Moreover, we model the end-to-end AI service process as three independent queues and characterize the latency agreement violation probability as a convex function via chernoff bound.

    \item We develop a corresponding algorithm to solve the proposed mechanism design problem. Specifically, we propose a mixed stable matching with contracts algorithm that jointly updates contract menus and matching decisions, and obtains market outcomes that account for both competitive allocation and asymmetric information. We further prove the existence of a mixed equilibrium, showing that the proposed dynamic matching-with-contracts formulation admits a well-defined market outcome.
    
    \item We extend the scope of the static matching-with-contracts model to a competitive and dynamic multi-principal market. Different from conventional settings where contract terms are fixed before matching, the proposed framework jointly considers contract menu design by competing principals, principal-agent matching, and market-state evolution. This captures the feedback between contract decisions and matching outcomes, where principals’ strategies influence agents’ choices, and agents’ choices in turn reshape the market environment faced by principals. 
    
    \item We conduct numerical simulation experiments to validate the effectiveness of the proposed framework and algorithm. Compared with benchmark schemes, the proposed method improves the total AI-RAN operator utility by at least $56.8\%$ and the social welfare by at least $51.7\%$ under representative congested-market settings.
\end{enumerate}

The remainder of this paper is organized as follows. Section \ref{sec:2} reviews the related literature and introduces the necessary background. Section \ref{sec:3} presents the framework illustration and system model. Section \ref{sec:4} presents the latency agreement violation model and the utility functions of AI users and AI-RAN operators. Section \ref{sec:5} formulates the incentive mechanism design problem under asymmetric information and competition. Section \ref{sec:6} develops the proposed algorithm and analyzes its computational complexity. Section \ref{sec:7} reports the simulation results. Finally, Section \ref{sec:8} concludes the paper.

\section{Related Works} \label{sec:2}
In this section, we review the literature related to this paper in two branches, AI-RAN service provisioning in Section \ref{sec:2.1} and utilizing matching theory and contract theory for incentive mechanism design in Section \ref{sec:2.2}.

\subsection{AI-RAN Service Provisioning} \label{sec:2.1}
Existing AI-RAN/O-RAN studies can be broadly understood from two perspectives. The first perspective is AI-for-RAN, where AI methods are introduced to improve network control and operational efficiency. In this direction, learning-based methods have been developed for RAN slicing, xApp design, real-time resource allocation, load balancing, data-driven configuration, and conflict resolution \cite{zhou2021ranslicing, kouchaki2022actor, rezazadeh2023madrloran, lai2023loadbalancingoran, mahmoud2024datadrivenoran, rezazadeh2024steporan}. These studies demonstrate the effectiveness of AI in optimizing conventional RAN functions under dynamic network conditions.

The second perspective is AI-on-RAN, where the RAN is treated as a distributed service platform for AI task execution. Compared with AI-for-RAN, research on AI-on-RAN is still at an infancy stage, although it is becoming increasingly important as AI services become more computation-intensive, latency-sensitive, and geographically distributed \cite{khan2023airan, bonati2021intelligence, giannopoulos2022supporting}. Existing studies have mainly investigated system architecture design, AI/ML workflow support, communication-computation resource orchestration, service placement, and performance optimization under network-side control \cite{bonati2021intelligence, giannopoulos2022supporting, filali2023commcomporan, abedin2022elastic}. These works provide important technical foundations for deploying AI services over programmable RAN infrastructures.

However, most existing studies abstract away the strategic behavior of participants and implicitly assume that resources can be coordinated once a network-side objective is specified. In practice, AI-on-RAN service provisioning may involve multiple self-interested AI-RAN operators and heterogeneous AI users. AI-RAN operators have different radio and computing capabilities, operating costs, and service qualities, while users may have private service valuations and latency sensitivities. Without an effective incentive mechanism, AI-RAN operators may lack motivation to provide AI services, and AI users may not be efficiently matched with suitable RAN/computation resources.

This gap motivates the market-oriented incentive design problem studied in this paper. Unlike existing AI-RAN/O-RAN works that mainly focus on network-side control and resource orchestration, we consider a competitive AI-on-RAN service market in which multiple AI-RAN operators design latency-price contract menus and compete for AI users under congestion-dependent service reliability.

\subsection{Matching Theory and Contract Theory for Incentive Design} \label{sec:2.2}
Contract theory is a mature tool for incentive mechanism design under asymmetric information, because a properly designed contract menu can induce agents with private types to reveal their preferences through self-selection \cite{li2023book}. Matching theory provides a complementary framework for decentralized market formation with heterogeneous agents and bilateral preferences \cite{lovasz2009matching}. The seminal work in \cite{hatfield2005matching} established the matching-with-contracts framework, showing that matching decisions and contract terms can be incorporated into a unified market model. These theoretical foundations are closely related to our problem, where AI users choose among AI-RAN operators and contract items, while AI-RAN operators design service terms under incomplete information.

In wireless and edge computing systems, contract theory and matching theory have been widely used to address resource trading, task assignment, and incentive compatibility. The authors in \cite{zhou2019vfccontractmatching} proposed a contract-matching approach for computation resource allocation and task assignment in vehicular fog computing. In \cite{su2020hierarchicalmec}, hierarchical multi-access edge computing (MEC) offloading was studied through contract design and Bayesian matching. The work in \cite{su2021mecmatchingcontracts} further adopted matching with contracts for resource trading and price negotiation in MEC. In addition, \cite{lim2021towards} designed a multi-dimensional contract-matching mechanism for federated learning in UAV-enabled Internet of Vehicles. These studies demonstrate that contract and matching tools can jointly capture incentive compatibility, decentralized association, and resource trading in edge computing environments.

These economic tools have also been introduced into edge intelligence services where incentive compatibility is essential. For federated learning, the authors in \cite{hui2022quality} developed a quality-aware incentive mechanism based on matching games. The authors in \cite{wang2024cvflcontract} designed contract-theoretic incentives for clustered vehicular federated learning. The authors in \cite{yao2025saflcontract} proposed a learning-based contract design method for semi-asynchronous federated learning. For AIGC services, the authors in \cite{ye2025aigccontract} studied contract-theoretic optimization supported by prompt engineering and edge computing. These works confirm the relevance of matching and contract theory for edge intelligence systems with heterogeneous participants, private information, and distributed resources.

Nevertheless, these studies cannot be directly applied to competitive AI-on-RAN service provisioning. First, many contract-theoretic works focus on a single principal or a single platform, while our problem involves multiple competing AI-RAN operators whose contract menus jointly determine user association and service congestion. Second, existing matching-based and contract-matching methods usually treat service quality, allocation, and payment as static terms, while AI-on-RAN service reliability depends on the traffic load induced by users' matching decisions.

In summary, our work differs from the existing literature in three main aspects. First, compared with AI-RAN/O-RAN architectural studies \cite{khan2023airan, bonati2021intelligence, giannopoulos2022supporting}, we focus on incentive design for a competitive AI-RAN service market. Second, compared with existing contract-theoretic and matching-based incentive mechanisms \cite{zhou2019vfccontractmatching, su2020hierarchicalmec, su2021mecmatchingcontracts, lim2021towards, hui2022quality, wang2024cvflcontract, yao2025saflcontract, ye2025aigccontract}, we jointly model competitive operator-side contract design, user-side matching, and congestion-dependent service reliability. Third, compared with classical matching-with-contracts theory \cite{hatfield2005matching}, we incorporate AI-RAN-specific queueing dynamics and operator-side contract menu update into the market formation process.

\section{Framework Illustration and System Model} \label{sec:3}
In this section, we first illustrate the framework of applying matching-with-contracts for the AI-on-RAN service market in Section \ref{sec:3.1}. Next, we present the system model of our proposed framework in Section \ref{sec:3.2}.

\begin{figure*}[!t]
    \centering
    \includegraphics[width=.95\linewidth]{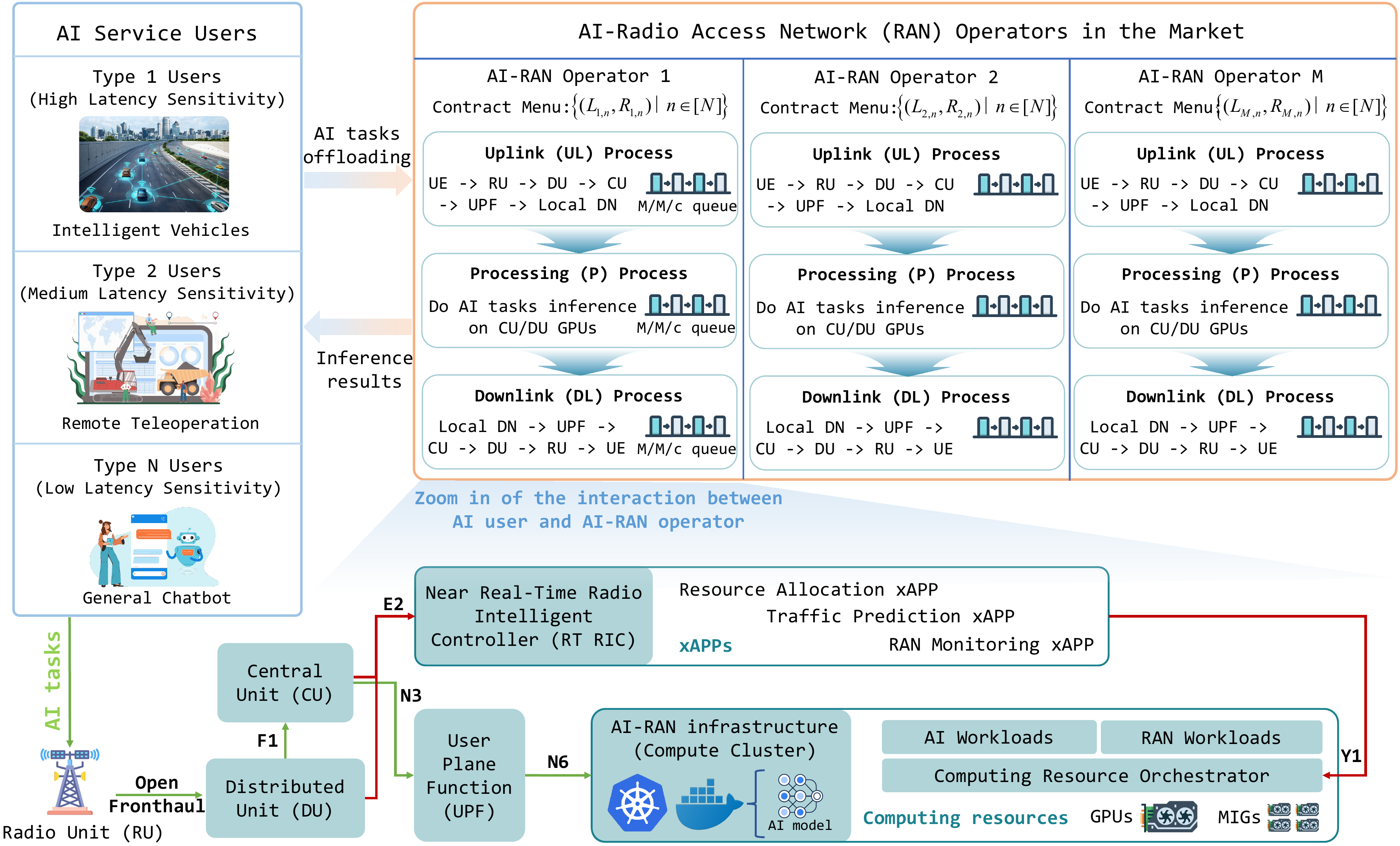}
    \caption{Framework illustration of the AI task offloading in a matching-with-contracts AI-RAN service market.}
    \label{fig:framework}
\end{figure*}
\subsection{Framework Illustration} \label{sec:3.1}
Fig.~\ref{fig:framework} illustrates the proposed matching-with-contracts for the AI-on-RAN service market framework, which captures the interaction between AI users and competing AI-RAN operators in a contract-driven service market. In this framework, AI users generate inference tasks and offload them to AI-RAN operators that integrate radio access and edge computing capabilities. We consider representative AI service scenarios, including intelligent vehicles, remote teleoperation, and general chatbot applications. However, the proposed framework is general and can be extended to other AI services.

Notably, the latency sensitivity of AI users is hidden information to the AI-RAN operators, resulting in information asymmetry. This information asymmetry hinders AI-RAN operators from formulating optimal pricing strategies. To address this challenge, AI-RAN operators leverage contract theory to design contract menus that induce truthful self-selection from AI users while maximizing AI-RAN operator profit \cite{li2023book}. Specifically, each AI-RAN operator broadcasts a contract menu consisting of latency agreements $\mathbf{L}$ and corresponding service prices $\mathbf{R}$ designed for different AI user types. Since multiple AI-RAN operators coexist in the market, each AI user selects the contract item that maximizes its utility, thereby implicitly determining the operator that will serve its tasks \cite{hatfield2005matching, rostek2019matching, macho2021agency}.

After contract selection, the AI task execution process consists of three sequential stages: uplink transmission, AI inference processing, and downlink transmission. These three stages jointly determine the end-to-end latency experienced by the AI user. To capture the impact of communication and computation resource contention, each stage is modeled as an M/M/c queue. This queueing abstraction enables tractable analysis of latency violation probability under varying traffic load and resource availability.

The zoom-in view in Fig.~\ref{fig:framework} illustrates the detailed interaction between an AI user and an AI-RAN operator. After selecting an AI-RAN operator, the AI user offloads its task from the user equipment (UE) through the uplink communication pipeline, which traverses the radio unit (RU), distributed unit (DU), centralized unit (CU), and user plane function (UPF) before reaching the AI-RAN operator-controlled computing cluster. The computing resources are co-located with the DU/CU to enable low-latency AI task inference. Moreover, GPU resources are partitioned into multiple virtual inference instances using techniques such as multi-instance GPU (MIG) \cite{kundu2025ai}, enabling fine-grained sharing between AI workloads and RAN workloads. The RU provides wireless connectivity, while the DU and CU perform baseband processing, scheduling, and protocol stack operations. These components are interconnected through standardized interfaces, including the Open Fronthaul interface between RU and DU, the F1 interface between DU and CU, and the N3 and N6 interfaces connecting the CU and UPF to the local data network hosting AI workloads. The AI inference task is executed on GPU resources, where AI models are deployed in containerized environments such as Docker. After inference, the results are transmitted back to the AI user through the downlink pipeline along the reverse path.

In addition to the data transmission pipeline, the AI-RAN system incorporates a closed-loop resource orchestration mechanism to coordinate resource allocation between RAN and AI workloads. The near-real-time radio intelligent controller (near-RT RIC) hosts multiple control applications (xApps) \cite{ko2024edgeric}, including traffic prediction, resource allocation, and network monitoring. These xApps continuously monitor network conditions and generate key performance indicator (KPI) metrics. Through the Y1 interface, these KPI metrics are delivered to the computing resource orchestrator, which dynamically determines the allocation of computing resources between RAN functions and AI inference workloads \cite{kundu2025ai, shah2025interplay}. This coordination mechanism improves computing resource utilization while ensuring communication reliability.

\subsection{System Model} \label{sec:3.2}
We consider an AI-on-RAN service market in which multiple AI service users request AI inference services from a set of competing AI-RAN operators. We utilize $\mathcal{S}_u=\{u_i \mid i \in [I]\}$ to denote the set of AI users, where $[I]=\{1,\dots,I\}$ and $u_i$ indicates the $i$-th AI user. We model the AI tasks generation process as a Poisson process, in which each AI service user $u_i$ generates AI tasks with rate $\delta$ (tasks/s). For brevity and without loss of generality, we assume the AI task generation rate is identical for AI users. We utilize a three-tuple $\{d_i, \tau, d_o\}$ to inscribe the AI task, where $d_i$ is the average AI task input size, $\tau$ is the computation workload of the AI task, and $d_o$ indicates the size of the AI service processed result. We utilize $\mathcal{S}_o=\{o_m \mid m \in [M]\}$ to represent the set of AI-RAN operators, where $[M]=\{1,\dots,M\}$ and $o_m$ indicates the $m$-th AI-RAN operator.

We assume the AI task is indivisible and therefore each AI task can only be offloaded to one AI-RAN operator. Regarding the AI task offloading process, each AI user $u_i$ selects the promised AI service among the contract menus provided by all AI-RAN operators. Akin to \cite{wen2024diffusion, zhan2025distributionally}, we model the contract menu provided by each AI-RAN operator $o_m$ in the form of \textbf{\{AI service latency agreement, AI service price\}}. Given that the AI-RAN operator might compete for the market share, the optimal contract menu of each operator is affected by other operators. Moreover, AI-RAN operators are heterogeneous in their available radio capacities and computing resources, which will also affect the contract menu formulated by the AI-RAN operator.

As we mentioned in Section \ref{sec:3.1}, the information asymmetry exists between AI-RAN operators and AI users. We consider leveraging contract theory to derive the optimal contract menu $\mathcal{C}_m$ for $o_m$. We categorize AI users in the market into $N$ types via data mining technologies \cite{lim2021towards}, which is defined as $\Theta=\{\theta_n \mid n \in [N]\}$. We use $\mathcal{I}_n$ to represent the set of type-$n$ AI users in the market, in which $\mathcal{S}_u = \{\mathcal{I}_n \mid n \in[N]\}$. We define the contract menu formulated by each AI-RAN operator $o_m$ as $\mathcal{C}_m = \{(L_{m,n},R_{m,n}) \mid n \in [N]\}$. Here, $L_{m,n}$ and $R_{m,n}$ mean the AI service latency agreement and AI service price for type-$n$ AI users designed by $o_m$, respectively. Considering the competition exist among AI-RAN operators, we utilize a binary matrix $\mathbf{A}=\{a_{m,n} \mid m \in [m], n \in [N]\}$ to represent the matching condition between $N$ types of AI users and $M$ AI-RAN operators. We consider $a_{m,n} \in \{0, 1\}$, in which $a_{m,n} = 1$ indicates the type-$n$ AI user is matched with $o_m$. Conversely, we set $a_{m,n} = 0$ when a type-$n$ AI user is not matched with $o_m$.

\section{Utility Functions of AI Users and AI-RAN Operators} \label{sec:4}
In this section, we present the AI task latency agreement violation model in Section \ref{sec:4.1}. Subsequently, we design the utility functions of AI-RAN operators and AI users in Section \ref{sec:4.2}.

\subsection{AI Task Latency Agreement Violation Model} \label{sec:4.1}
As per Fig.~\ref{fig:framework}, we consider each AI task experiences uplink, processing, and downlink three stages, which are collectively determine the end-to-end service latency perceived by the AI user. Notably, since $N$ types of AI users exist in the system, each type of AI users have different latency requirement, we consider adopting a priority scheduling policy to inscribe the latency model of each type of AI users. We adopt the preemptive-resume priority policy to construct the latency model of each type of AI users.

We use $T_{m,n}^{\mathrm{UL}}$, $T_{m,n}^{\mathrm{P}}$, and $T_{m,n}^{\mathrm{DL}}$ to denote the sojourn times of the uplink, processing, and downlink stages, respectively, when a type-$n$ AI task is served by operator $o_m$. Together, these three stages characterize the end-to-end AI service latency in AI-RAN. We model the total latency perceived by the type-$n$ AI user as
\begin{equation} \label{eq:1}
    T_{m,n} = T_{m,n}^{\mathrm{UL}} + T_{m,n}^{\mathrm{P}} + T_{m,n}^{\mathrm{DL}}.
\end{equation}
To analytically characterize the AI service latency violation probability, we model each stage as an independent M/M/$c$ queue \cite{bi2024cost, 3gpp, liao2021adaptive}. Concretely, we model the AI task uplink transmission process on the AI-RAN operator $o_m$ as an $M/M/c_m^{\mathrm{UL}}$ queue. In AI-RAN, $c_m^{\mathrm{UL}}$ represents the effective amount of uplink resources allocated to AI service traffic, measured in terms of parallel service capacity (e.g., Physical Resource Blocks or sub-channel groups). This allocation is governed by near-RT RIC, typically enforced by near-RT RIC–hosted xApps that perform dynamic scheduling and resource partitioning based on current traffic load and service requirements, while higher-level policies (e.g., AI model updates) may be configured by Non-RT RIC via rApps on a slower timescale \cite{kundu2025ai, shah2025interplay, polese2025beyond}.

For the latency analysis within one contract-design horizon, we assume $c_m^{\mathrm{UL}}$ is quasi-static. We utilize $D_m^{\mathrm{UL}}$ to represent the uplink transmission data rate between $u_i$ and $o_m$ on one sub-channel. We define the uplink service rate of one sub-channel of the AI-RAN operator $o_m$ as
\begin{equation} \label{eq:2}
    \mu_{m}^{UL} \triangleq  \frac{D_{m}^{UL}}{d_i}.
\end{equation}
We model the aggregate uplink traffic arrival rate at the AI-RAN operator $o_m$ for the type-$n$ AI user as
\begin{equation} \label{eq:3}
    \lambda_{m,n}^{\mathrm{UL}} = \lambda_{m,\leq n}^{\mathrm{UL}} \triangleq \sum_{j = 1}^{n} {|\mathcal{I}_j| a_{m,j} \delta},
\end{equation}
which indicates that only the AI users with higher priority will affect the AI task arrival rate for type-$n$ AI users.

Under the assumed $M/M/c$ uplink queuing model, we characterize the uplink sojourn time distribution for type-$n$ AI task at AI-RAN operator $o_m$ as
\begin{equation} \label{eq:4}
    \begin{aligned}
        \mathbb{P}\!\left(T_{m,n}^{\mathrm{UL}} > t\right) = & \left(1 + \frac{P_{m,n}^{\mathrm{UL}} \mu_m^{\mathrm{UL}}}{r_{m,n}^{\mathrm{UL}} - \mu_m^{\mathrm{UL}}}\right)  e^{-\mu_m^{\mathrm{UL}} t} \\ 
        &- \frac{P_{m,n}^{\mathrm{UL}} \mu_m^{\mathrm{UL}}}{r_{m,n}^{\mathrm{UL}} - \mu_m^{\mathrm{UL}}}
        e^{- r_{m,n}^{\mathrm{UL}} t},
    \end{aligned}
\end{equation}
where $P_{m,n}^{\mathrm{UL}}$ denotes the probability that an arriving AI service request experiences uplink queueing delay for type-$n$ AI users, which is given by the Erlang--C formula \cite{takacs1969erlang}. The parameter $r_{m,n}^{\mathrm{UL}}$ represents the excess uplink service capacity for type-$n$ AI users, defined as
\begin{equation} \label{eq:5}
    r_{m,n}^{\mathrm{UL}} = c_m^{\mathrm{UL}} \mu_m^{\mathrm{UL}} - \lambda_{m,n}^{\mathrm{UL}}.
\end{equation}
To ensure queue stability, we assume $r_{m,n}^{\mathrm{UL}} > 0$.

Analogously, we model the AI task computation process on $o_m$ as an $M/M/c_m^{\mathrm{P}}$ queue, where $c_m^{\mathrm{P}}$ denotes the number of GPU instances available to serve incoming AI tasks. Since baseband processing and RAN control functions are assigned higher execution priority, the computation capacity available to AI task inference is inherently constrained to the prevailing RAN workload \cite{kundu2025ai, shah2025interplay,polese2025beyond}. For analytical tractability, we consider a quasi-static control snapshot over the contract-design horizon and assume $c_m^{\mathrm{P}}$ remains fixed. We use $D_m^{\mathrm{P}}$ to represent the computation capacity of one instantiated GPU inference unit. We define the computation service rate of one instanced GPU of the AI-RAN operator $o_m$ as
\begin{equation} \label{eq:6}
    \mu_{m}^{P} \triangleq  \frac{D_{m}^{P}}{\tau}.
\end{equation}
We define the AI task arrival rate for type-$n$ AI users as $\lambda_{m,n}^{P} = \lambda_{m, n}^{\mathrm{UL}}$. Akin to the definition in \eqref{eq:4}, we define the processing time distribution at AI-RAN operator $o_m$ as
\begin{equation} \label{eq:7}
    \begin{aligned}
        \mathbb{P}\!\left(T_{m,n}^{\mathrm{P}} > t\right) = & \left(1 + \frac{P_{m,n}^{\mathrm{P}} \mu_m^{\mathrm{P}}}{r_{m,n}^{\mathrm{P}} - \mu_m^{\mathrm{P}}}\right)  e^{-\mu_m^{\mathrm{P}} t} \\ 
        &- \frac{P_{m,n}^{\mathrm{P}} \mu_m^{\mathrm{P}}}{r_{m,n}^{\mathrm{P}} - \mu_m^{\mathrm{P}}}
        e^{- r_{m,n}^{\mathrm{P}} t}.
    \end{aligned}
\end{equation}

We model the downlink process similarly to the uplink process, and therefore, we define the downlink sojourn time distribution at the AI-RAN operator $o_m$ as
\begin{equation} \label{eq:8}
    \begin{aligned}
        \mathbb{P}\!\left(T_{m,n}^{\mathrm{DL}} > t\right) = & \left(1 + \frac{P_{m,n}^{\mathrm{DL}} \mu_m^{\mathrm{DL}}}{r_{m,n}^{\mathrm{DL}} - \mu_m^{\mathrm{DL}}}\right)  e^{-\mu_m^{\mathrm{DL}} t} \\ 
        &- \frac{P_{m,n}^{\mathrm{DL}} \mu_m^{\mathrm{DL}}}{r_{m,n}^{\mathrm{DL}} - \mu_m^{\mathrm{DL}}}
        e^{- r_{m,n}^{\mathrm{DL}} t}.
    \end{aligned}
\end{equation}
Here, $\mu_{m}^{DL}$ is akin to the definition in \eqref{eq:2}, which is defined as $\mu_{m}^{DL} \triangleq  \frac{D_{m}^{DL}}{d_o}$. We utilize $D_{m}^{DL}$ to represent the downlink data transmission rate between $u_i$ and $o_m$ on one sub-channel. We set $c_{m}^{DL}$ as the number of sub-channels possessed by the AI-RAN operator $o_m$.

We assume the three stages of AI tasks processing on the AI-RAN operator $o_m$ are independent, and therefore define the latency agreement violation probability distribution for type-$n$ AI tasks at $o_m$ as
\begin{equation} \label{eq:9}
     \mathbb{P}\!\left(T_{m,n} > t\right) = \mathbb{P}\!\left( (T_{m,n}^{\mathrm{UL}} + T_{m,n}^{\mathrm{P}} + T_{m,n}^{\mathrm{DL}}) > t\right).
\end{equation}
The expression in \eqref{eq:9} corresponds to a hypoexponential tail probability. While exact, it leads to a non-convex and analytically cumbersome form, which makes it difficult to embed into the subsequent contract optimization and equilibrium analysis. Therefore, we apply the Chernoff bound \cite{hellman1970probability} to derive a convex upper-bound approximation of the latency agreement violation probability for type-$n$ AI tasks served by $o_m$, which is defined as
\begin{equation} \label{eq:10}
    \tilde{\mathbb{P}}\!\left(T_{m,n} > t\right) = e^{-\eta t}\prod_{s}{G_{m,n}^{s}(\eta)}.
\end{equation}
Here, we consider $s \in \{UL, P, DL\}$. We define the function $G_{m,n}^{s}(\eta)$ as
\begin{equation} \label{eq:11}
    G_{m,n}^s(\eta) = \left( {\left(1 - P_{m,n}^s\right) + P_{m,n}^s\frac{{r_{m,n}^s}}{{r_{m,n}^s - \eta}}} \right) \cdot \frac{{\mu _m^s}}{{\mu _m^s - \eta}},
\end{equation}
in which $\eta$ is defined as
\begin{equation} \label{eq:12}
    \eta = \zeta \cdot \min_s\{\mu_{m,n}^s-\frac{\lambda_{m,n}^{s}}{c_{m,n}^{s}}\},
\end{equation}
where $\zeta \in (0, 1)$.

For clarity, in the remainder of this paper, we leverage $\tilde{p}(t)$ to represent the equation \eqref{eq:10}. For instance, we utilize $\tilde{p}(L_{m,n})$ to denote the latency agreement violation probability of AI-RAN operator $o_m$ regarding the type-$n$ contract item.

\subsection{Utility Functions} \label{sec:4.2}

\subsubsection{AI-RAN Operator Utility}
In this paper, we consider the utility of the AI-RAN operator is composed of the payment received from matched AI users, the penalty regarding the latency service agreement violation, and the energy cost of running the AI model. Since we assume each AI-RAN operator $o_m$ designs a contract menu $\mathcal{C}_m = \{(L_{m,n},R_{m,n}) \mid n \in [N]\}$ for $N$ types of AI users, we design the utility function of AI-RAN operator $o_m$ as follows
\begin{equation} \label{eq:13}
    \begin{aligned}
        \pi _{op}^m({{\cal C}_m};{{\cal C}_{-m}}) = \sum_{n=1}^{N}{|{{\cal I}_n}|{a_{m,n}} \delta \left( R_{m,n} - \bar{C} \tilde{p}\left( {L_{m,n}} \right) \right)} - \varrho E_m.
    \end{aligned}
\end{equation}
Here, ${\cal C}_{-m}$ indicates the contract menus designed by other AI-RAN operators besides the AI-RAN operator $o_m$, which will affect the contract menu design of $o_m$. $\bar{C}$ captures the ai-ran operator-side expected economic loss caused by violating the posted latency agreement, including compensation cost, reputation loss, and operational overhead. $\varrho$ is the conversion coefficient between the energy cost and the US dollar. $E_m$ is the energy cost of $o_m$ in processing the AI tasks that matched AI users offloaded, which is defined as
\begin{equation} \label{eq:14}
    E_m =  \sum_{n=1}^{N}{|{{\cal I}_n}|{a_{m,n}} \delta} \epsilon_m \phi_m.
\end{equation}
$\epsilon_m$ represents the energy cost per Floating-point OPeration (FLOP) of the AI-RAN operator $o_m$. $\phi_m$ denotes the model complexity, which is inscribed in the number of FLOP used by $o_m$ for running the AI model.

\subsubsection{AI User Utility and Hidden Information}
We consider that the utility of the AI user is affected by the AI service quality, AI service latency, the AI model fee charges by the AI-RAN operator, and the potential penalty refund received from the AI-RAN operator \cite{liu2024deep, kang2019toward, zhan2025learning, li2021contract}. Therefore, we model the utility function of the type-$n$ AI user as
\begin{equation} \label{eq:15}
    \pi_u(L_{m,n},R_{m,n};\theta_n) = \alpha_x q_m - \beta_y L_{m,n} - R_{m,n} + \bar{R} \tilde{p}(L_{m,n}).
\end{equation}
Here, $\alpha_x$ denotes the sensitivity of the AI user $u_i$ regarding the AI service quality.
$\beta_y$ indicates the sensitivity of the AI user $u_i$ regarding the AI service latency. $\bar{R}$ is the penalty fee that AI-RAN operator $o_m$ pays to the AI user when the latency agreement is violated, which is a constant.

Notably, the parameters $\alpha_x$ and $\beta_y$ of AI users are private information to the AI-RAN operator, which might hinder the optimal contract menu derivation for the AI-RAN operator. To this end, we intend to utilize contract theory \cite{li2023book}, which possesses the truthful revealing property, to assist the AI-RAN operator in deriving the optimal contract menu under hidden information. We illustrate the connection between the AI user's type, $\theta_n$, and the private parameters $\alpha_x$ and $\beta_y$ of AI users in Section \ref{sec:5.1}.

\section{Problem Formulation} \label{sec:5}
In this section, we first present the contract menu design without competition in Section \ref{sec:5.1}. Next, we reformulate the contract menu design problem in Section \ref{sec:5.2}. Lastly, we illustrate the contract menu design problem under competition in Section \ref{sec:5.3}.

\subsection{Contract Design Under Asymmetric Information} \label{sec:5.1}

In this section, we formulate the contract design problem of the AI-RAN operator $o_m$ under asymmetric information, in which the competition among AI-RAN operators is omitted for clarity.

\subsubsection{Two-Dimensional Private Information and Contract Menu}
Observing \eqref{eq:15}, AI users are characterized by a private type $(\alpha_x,\beta_y)$. By referring to \cite{lim2021towards, xiong2020multi, zhan2025vision}, we assume the AI-RAN operator $o_m$ can categorize the private type of AI users into the set $\{\alpha_x \mid x \in [X]\}$ and the set $\{\beta_y \mid y \in [Y]\}$ via data mining technologies. We sort the private types as follows:
\begin{equation} \label{eq:16}
    \alpha_1 \le \alpha_2 \le \cdots \le \alpha_X,
    \qquad
    \beta_1 \ge \beta_2 \ge \cdots \ge \beta_Y.
\end{equation}
Therefore, we define the contract menu provided by the AI-RAN operator $o_m$ as
\begin{equation} \label{eq:17}
    \mathcal{C}_m = \{(L_{m,x,y}, R_{m,x,y}) \mid x\in [X],\, y\in [Y]\},
\end{equation}
where $(L_{m,x,y}, R_{m,x,y})$ denotes the latency agreement and the corresponding price designed for the type-$(\alpha_x,\beta_y)$ AI user. Different from the definition of the AI user's utility function in \eqref{eq:15}, we re-express the utility function of AI users in an explicit manner
\begin{equation} \label{eq:18}
    \begin{aligned}
        \pi_u(L_{m,x,y}, R_{m,x,y};\alpha_x, \beta_y)
        & =  \alpha_x q_m - \beta_y L_{m,x,y} \\
        &- R_{m,x,y} + \bar R \tilde p(L_{m,x,y}).
    \end{aligned}
\end{equation}

\subsubsection{Impact of Private Information on Contract Selection} To maximize the AI-RAN operator's utility under information asymmetry, the contract menu should satisfy the standard incentive compatibility (IC) and individual rationality (IR) constraints \cite{li2023book}.

\emph{Incentive compatibility} requires that each AI user maximizes its utility by selecting the contract item designed for its own type, i.e., the self-revealing principle, which is defined as
\begin{equation} \label{eq:19}
\begin{aligned}
    \pi_u(L_{m,x,y}, R_{m,x,y}; \alpha_x, \beta_y)
    \ge
    \pi_u(L_{m,x',y'}, R_{m,x',y'}; \alpha_x, \beta_y),
    \\
    \forall (x',y')\neq(x,y),\ \forall x,y.
\end{aligned}
\end{equation}
When $XY \times (XY-1)$ IC constraints \eqref{eq:19} are satisfied, all types of rational AI users will select the contract tailored for them, so as to maximize their utility.

\emph{Individual rationality} requires that each AI user obtains a non-negative utility by participating in the contract:
\begin{equation} \label{eq:20}
    \pi_u(L_{m,x,y}, R_{m,x,y}; \alpha_x, \beta_y) \ge 0, \quad \forall x,y.
\end{equation}
When $XY$ IR constraints \eqref{eq:20} are satisfied, all types of rationale AI users will have an incentive to select a contract and pay for the AI service provided by the AI-RAN operator.

\subsubsection{Reduction of the Contract Menu}
We assess how the two private parameters $(\alpha_x,\beta_y)$ affect AI
users’ contract selection behavior. Consider two arbitrary contract items
$L_{m,x,y}, R_{m,x,y}$ and $L_{m,x',y'}, R_{m,x',y'}$, the utility difference for the type-$(\alpha_x,\beta_y)$ AI user is defined as 
\begin{equation} \label{eq:21}
    \begin{aligned}
        &\pi_u(L_{m,x,y}, R_{m,x,y};\alpha_x,\beta_y) - \pi_u(L_{m,x',y'}, R_{m,x',y'};\alpha_x,\beta_y) \\
        &= -\beta_y(L_{m,x,y}-L_{m,x',y'}) - (R_{{m,x,y}}-R_{m,x',y'}) \\
        &+ \bar R\big( \tilde{p}(L_{{m,x,y}}) - \tilde{p}(L_{m,x',y'}) \big),
    \end{aligned}
\end{equation}
which is independent of $\alpha_x$. Therefore, although AI users possess two-dimensional private information, only $\beta_y$ plays a role in affecting the IC constraints. All users with the same latency sensitivity $\beta_y$ will select the same contract item, regardless of their quality sensitivity $\alpha_x$.

Observing IR constraints in \eqref{eq:20}, when we fix the private parameter $\beta_y$, the AI user's utility is increasing in $\alpha_x$. The IR constraint is most restrictive for the smallest quality sensitivity $\alpha_1$ when we fix $\beta_y$. Therefore, we can reduce the $XY$ IR constraints in \eqref{eq:20} to
\begin{equation} \label{eq:22}
    \pi_u(L_{m,1,y}, R_{m,1,y}; \alpha_1, \beta_y) \ge 0, \quad \forall y,
\end{equation}
while guaranteeing the IR constraints are satisfied for all types of AI users.

Upon the above observations, both IC and IR constraints of AI users are solely affected by the private parameter $\beta_y$ with fixed $\alpha_1$. Therefore, we can reduce the original two-dimensional contract menu to a one-dimensional menu indexed only by $\beta_y$:
\begin{equation} \label{eq:23}
    \mathcal{C}_m = \{(L_{m,y}, R_{m,y}) \mid y \in [Y] \}.
\end{equation}

For notational clarity and consistency with the system model defined in Section \ref{sec:3.2}, we utilize $\Theta = \{\theta_n \mid n \in [N]\}$ to represent the set of $N$ types of AI users in the market. Here, we consider $N = Y$ and $\theta_n \triangleq (\beta_y)$. We set the contract menu formulated by the AI-RAN operator $o_m$ as $\mathcal{C}_m=\{(L_{m,n}, R_{m,n}) \mid n \in [N]\}$.

We formulate the contract design problem of the AI-RAN operator $o_m$ as
\begin{subequations} \label{eq:24}
    \begin{align}
        \max_{{\cal C}_m} \quad & \sum_{n=1}^{N} |{\cal I}_n|\delta \left( R_{m,n} - \bar C \tilde{p}(L_{m,n}) - \varrho \epsilon_m \phi_m \right) \label{eq:24a}\\
        \text{s.t.}\quad
        &
        \pi_u(L_{m,n},R_{m,n};\theta_n) \geq 0, \quad \forall n, \label{eq:24b}
        \\
        &
        \pi_u(L_{m,n},R_{m,n};\theta_n) \geq  \pi_u(L_{m,n'},R_{m,n'};\theta_n), \notag
        \\ & \forall n' \neq n, \forall n. \label{eq:24c}
    \end{align}
\end{subequations}
Since we do not consider the competition among AI-RAN operators, the optimization objective function in \eqref{eq:24a} is different from the utility function of the AI-RAN operator defined in \eqref{eq:13}. Without competition, we assume all of the AI users request AI services from the AI-RAN operator $o_m$, i.e., $\{a_{m,n} = 1 \mid n \in [N]\}$. Equations \eqref{eq:24b} and \eqref{eq:24c} are IR and IC constraints of AI users, respectively. Regarding the private parameter $\alpha_x$, we only consider the case of $\alpha_1$.

\subsection{Contract Design Problem Reformulation} \label{sec:5.2}
Although the IC and IR constraints in \eqref{eq:24b} and \eqref{eq:24c} are convex, the number of constraints, $N^2$, grows quadratically with the number of types. Moreover, the coupled structure of the IC constraints further exacerbates the difficulty of analytically deriving the optimal contract. Similar to previous works \cite{zhan2025distributionally, lim2021towards, li2021contract}, we aim to reduce the original $N^2$ constraints to $N$ binding constraints by conducting IR and IC constraints reduction. To this end, we enable closed-form elimination of reward variables, $\{R_{m,n} \mid n \in [N]\}$.

\subsubsection{Constraints Reduction and Reward Variables Elimination} After the contract menu reduction, the optimization problem in \eqref{eq:24} is a single dimensional contract screening problem. By referring to the proof skeleton in \cite{zhan2025distributionally, lim2021towards, li2021contract, xiong2020multi}, we can deduce the following sufficient conditions for the feasible contract.
\begin{proposition} [Conditions for Contract Feasibility] \label{prop:1}
    For contract feasibility, we can reduce the IC and IR constraints in \eqref{eq:24b} and \eqref{eq:24c} as follows:
    \begin{enumerate}
        \item [1).] $L_{m,1} \le L_{m,2} \le \cdots \le L_{m,N}$,
        \item [2).] $\pi_u(L_{m,1},R_{m,1};\theta_1) \ge 0$,
        \item [3).] $\pi_u(L_{m,n},R_{m,n};\theta_n) \ge \pi_u(L_{m,n-1},R_{m,n-1};\theta_n), \; \forall n \in \{2,\cdots,N\}.$
        \item [4).] $\pi_u(L_{m,n},R_{m,n};\theta_n) \ge \pi_u(L_{m,n+1},R_{m,n+1};\theta_n), \; \forall n \in \{1,\cdots,N-1\}.$
    \end{enumerate}
\end{proposition}
The first monotonicity condition is the necessary condition for the feasible contract derivation. The second, third, and fourth conditions are the reduced IR and IC constraints, respectively.

With Proposition \ref{prop:1}, we can derive a closed-form representation of reward variables in terms of the latency agreement variables. Concretely, at the optimal solution to \eqref{eq:24}, the worst-type IR constraint and all adjacent IC constraints bind, which will have the impact of only increasing the utility of the AI-RAN operator. We can re-state the optimal reward variables $\{R_{m,n} \mid n \in [N]\}$ via Proposition \ref{prop:2}.

\begin{proposition}[Optimal Reward Variables]
\label{prop:2}
For a set of latency agreement variables $\mathbf{L}_m$ that satisfies $L_{m,1} \le L_{m,2} \le \cdots \le L_{m,N}$ in a feasible contract menu $\mathcal{C}_m$, the optimal reward variables is defined as
\begin{equation}\label{eq:25}
R_{m,n}=
\begin{cases}
\alpha_1 q_m-\beta_1 L_{m,1}+\bar R\,\tilde p(L_{m,1}), & n=1,\\[4pt]
R_{m,n-1}
-\beta_n\!\left(L_{m,n}-L_{m,n-1}\right)
 \\
 +\bar R\!\left(\tilde p(L_{m,n})-\tilde p(L_{m,n-1})\right), & n\in\{2,\dots,N\}.
\end{cases}
\end{equation}
\end{proposition}
\subsubsection{Transformed Optimization Problem}
Upon Proposition~\ref{prop:2}, we can reduce the $2N$-variable problem \eqref{eq:24} to an $N$-variable optimization in $\mathbf{L}_{m} = \{L_{m,n} \mid n \in [N]\}$:
\begin{subequations} \label{eq:26}
    \begin{align}
        \max_{\mathbf{L}_m}
        \quad & \sum_{n=1}^{N} |{\cal I}_n| \delta \left( R_{m,n}(\mathbf{L}_m) -\bar C\,\tilde{p}(L_{m,n}) - \varrho \epsilon_m\phi_m \right) \label{eq:26a} \\
        \text{s.t.}\quad
        & L_{m,1} \le L_{m,2} \le \cdots \le L_{m,N}, \label{eq:26b}
    \end{align}
\end{subequations}
where $R_{m,n}(\mathbf{L}_m)$ is given by \eqref{eq:25}. We will utilize the problem \eqref{eq:26} as the basis for the subsequent competitive contract design.

\subsection{Contracts Design under Competition} \label{sec:5.3}
In this section, we extend the contract design to a competitive market with $M$ AI-RAN operators.

\subsubsection{Congestion-Dependent User Selection Equilibrium}
Under competition, each AI-RAN operator's utility is affected by the AI user's selection $\mathbf{A}$. Meanwhile, the latency agreement violation probability $\tilde{p}(\cdot)$ also affects both the AI user's utility and AI-RAN operator's utility, which depends on the AI task traffic load $\Lambda_m$ induced by $\mathbf{A}$. Here, we define $\Lambda_m$ as $\{\lambda_{m, n} \mid n \in [N]\}$, in which $\lambda_{m, n}$ is the same as the definition in \eqref{eq:3}.

Given a contract menu profile $\mathcal{S}_c=\{\mathcal{C}_m \mid m \in [M]\} = \{\mathcal{C}_m, \mathcal{C}_{-m}\}$, the utility of a type-$n$ AI user under operator $o_m$ is 
\begin{equation} \label{eq:27}
\begin{aligned}
    \pi_u(\mathcal{C}_m;\theta_n, \mathcal{C}_{-m}) = & \alpha_1 q_m
    -\beta_n L_{m,n} \\ 
    & -R_{m,n} +\bar R\,\tilde{p}(L_{m,n};\lambda_{m, n}(\mathbf{A})).
\end{aligned}
\end{equation}
Different from the canonical IC constraints \eqref{eq:19} that compare different types of AI users' utilities under fixed $\tilde{p}(\cdot)$, in \eqref{eq:27}, $\tilde{p}(\cdot)$ is induced by the AI users' selection matrix. For a given menu profile $\mathcal{S}_c$, we define the market outcome as a \emph{congestion-dependent selection equilibrium}.

\begin{definition}[Congestion-Dependent Selection Equilibrium]
\label{def:selection_equilibrium}
Given a menu profile $\mathcal{S}_c$, an AI user selection matrix $\mathbf{A}$ is a selection equilibrium if, for every type $n$, whenever $a_{m,n}=1$,
\begin{equation} \label{eq:28}
    \pi_u(\mathcal{C}_m;\theta_n, \mathcal{C}_{-m}) \ge \pi_u(\mathcal{C}_{m'};\theta_n, \mathcal{C}_{-m'}),\quad \forall m'\neq m,
\end{equation}
and $\pi_u(\mathcal{C}_m;\theta_n, \mathcal{C}_{-m})\ge 0$. We denote the equilibrium correspondence by $\mathbf{A}\in\Psi(\mathcal{S}_c)$.
\end{definition}

Notably, a larger market share for $o_m$ will increase $\lambda_{m, n}(\mathbf{A})$, which in turn increases $\tilde{p}(\cdot;\lambda_{m, n}(\mathbf{A}))$ and may reduce both the AI user utility and AI-RAN operator profit. To this end, the user selection matrix $\mathbf{A}$ is coupled with the optimal contract menu profile $\mathcal{S}_c$.

\subsubsection{Competitive Contract Design as an Equilibrium-Constrained Game}
By incorporating the congestion dependence, we re-state the utility function of operator $o_m$ under a contract menu profile $\mathcal{S}_c$ and an equilibrium allocation $\mathbf{A}\in\Psi(\mathcal{S}_c)$ as
\begin{equation} \label{eq:29}
    \begin{aligned}
        \pi_{op}^m({\mathcal C}_m;{\mathcal C}_{-m})
        &= \sum_{n=1}^{N} \lambda_{m, n}(\mathbf{A}) \Big(
        R_{m,n} - \\ 
        & \qquad \bar{C} \tilde{p}\left(L_{m,n};\lambda_{m, n}(\mathbf{A})\right) - \varrho \epsilon_m \phi_m \Big).
    \end{aligned}
\end{equation}.

Since $R_{m,n}$ can be expressed in closed form by \eqref{eq:25} as a function of $\mathbf{L}_m$, we parameterize the strategy of operator $o_m$ by $\mathbf{L}_m$ and write $\mathcal{C}_m=\mathcal{C}_m(\mathbf{L}_m)$.
Given competitors' strategies $\mathbf{L}_{-m}$, the best-response problem of operator $o_m$ can be written as an equilibrium-constrained program:
\begin{subequations} \label{eq:30}
    \begin{align}
        \max_{\mathbf{L}_m}\quad & \pi _{op}^m({{\cal C}_m}; {{\cal C}_{-m}}) \label{eq:30a} \\
        \text{s.t.}\quad
        & L_{m,1} \le L_{m,2} \le \cdots \le L_{m,N}, \label{eq:30b} \\
        & \mathbf{A}\in\Psi\!\big(\mathcal{C}_m(\mathbf{L}_m),\mathcal{C}_{-m}(\mathbf{L}_{-m})\big), \label{eq:30c} \\
        & a_{m,n}\in\{0,1\},\ \sum_{m=1}^M a_{m,n}\le 1,\ \forall n\in[N]. \label{eq:30d}
    \end{align}
\end{subequations}
Here, \eqref{eq:30b} is the monotonicity constraint, which is identical to \eqref{eq:26b}. Furthermore, constraint \eqref{eq:30c} ensures that the market share variables, $ \mathbf{A} $, are consistent with AI users’ utility-maximizing choices under congestion, rendering the contract design problem an equilibrium-constrained optimization. Notably, the correspondence $\Psi(\mathcal{S}_c)$ may admit multiple equilibrium allocations. Finally, constraint \eqref{eq:30d} defines the feasible space of AI users' contract selection.

Compared with \eqref{eq:26}, the key difference is that the effective market share, $\mathbf{A}$, should be consistent with users' equilibrium selections under congestion (Definition~\ref{def:selection_equilibrium}), and the violation probability for the type-$n$ AI task is evaluated at the induced load $\lambda_{m, n}(\mathbf{A})$.

The interactions among AI-RAN operators form a noncooperative game where each AI-RAN operator chooses $\mathbf{L}_m$ (equivalently $\mathcal{C}_m$) to maximize \eqref{eq:30a} subject to \eqref{eq:30b}--\eqref{eq:30d}. We define the competitive outcome as follows.

\begin{definition}[Competitive Contract Equilibrium] \label{def:competitive_equilibrium}
A strategy profile $\mathcal{L}^\star=\{ \mathbf{L}_m^\star \mid m \in [M] \}$ is a competitive contract equilibrium if for every AI-RAN operator $o_m$, $\mathbf{L}_m^\star$ solves \eqref{eq:30} given $\mathbf{L}_{-m}^\star$, and the corresponding allocation $\mathbf{A}^\star$ satisfies $\mathbf{A}^\star\in\Psi\!\big(\mathcal{C}_m(\mathbf{L}_m^\star), \mathcal{C}_{-m}(\mathbf{L}_{-m}^\star)\big)$. The equilibrium contract menu of operator $o_m$ is then $\mathcal{C}_m^\star=\mathcal{C}_m(\mathbf{L}_m^\star)$.
\end{definition}

\section{Proposed Algorithm} \label{sec:6}
In this section, we analyze the problem in Section \ref{sec:6.1}. Next, we present the proposed algorithm in Section \ref{sec:6.2}. Finally, we discuss the algorithm complexity in Section \ref{sec:6.3}.

\subsection{Problem Analysis} \label{sec:6.1}
Observing problem \eqref{eq:30}, the contract menu derivation of AI-RAN operators is coupled with the matching decisions of AI users via the congestion-dependent latency agreement violation probability $\tilde{p}(\cdot;\lambda_{m,n})$. We present the coupling structure of problem \eqref{eq:30} in Fig.~\ref{fig:problem_structure}. If we directly solve \eqref{eq:30} with a binary selection matrix $\mathbf{A}$, the market outcome is discrete and non-smooth, which poses a significant challenge for us to reach the equilibrium point. Specifically, a small variation in one AI-RAN operator's contract menu may trigger a re-selection of one AI user type. The AI users' utility will change accordingly since $\lambda_{m,n}$ and the corresponding latency agreement violation probability $\tilde{p}$ are altered along with $\mathbf{A}$. To this end, AI users need to re-select the AI-RAN operator, and AI-RAN operators need to re-design the contract menu.

\begin{figure}
    \centering
    \includegraphics[width=0.98\linewidth]{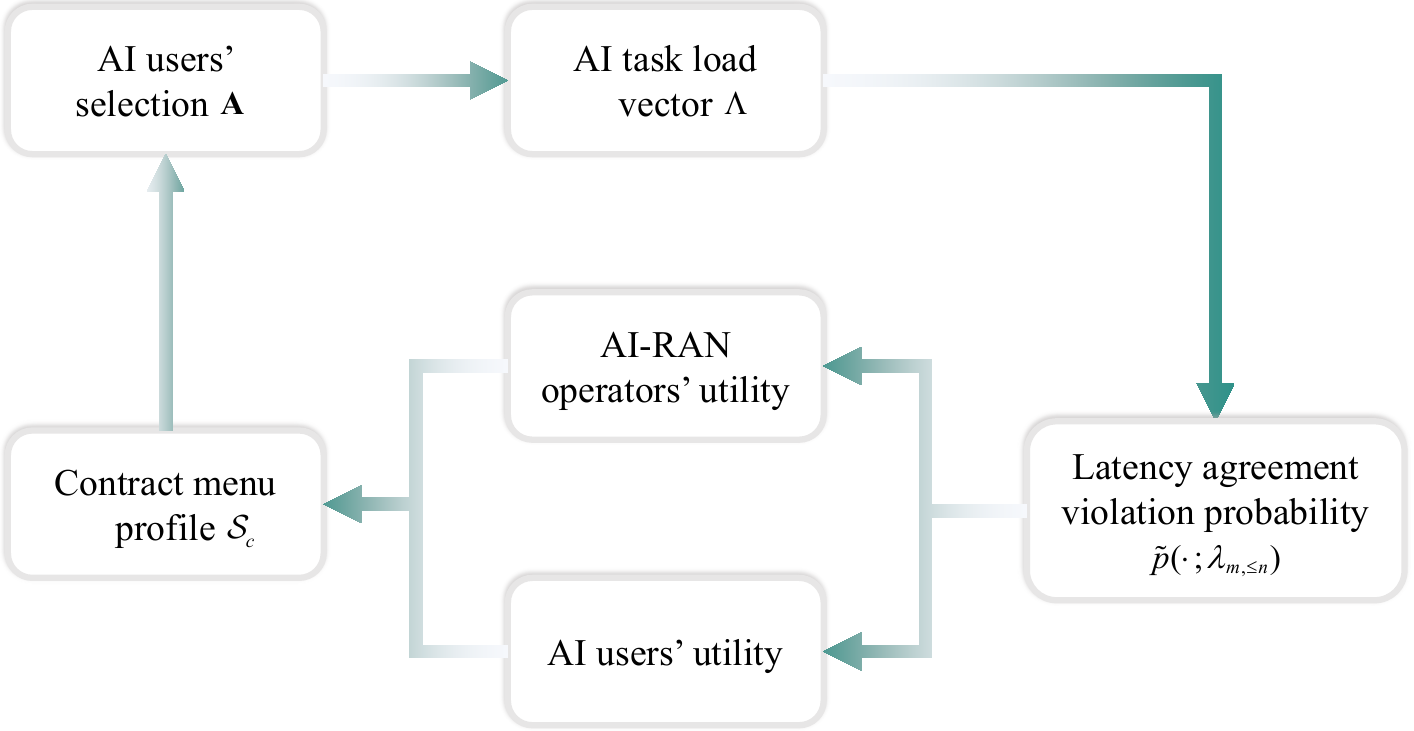}
    \caption{The illustration of the coupling structure of the problem \eqref{eq:30}.}
    \label{fig:problem_structure}
\end{figure}

To address this challenge, we propose a mixed stable matching with contracts algorithm, the design principle is inspired by \cite{hatfield2005matching, kelley1995iterative}. Our proposed algorithm is executed in the manner of the fixed-point iteration algorithm. Firstly, at the beginning of each iteration, each type-$n$ AI user is allowed to probabilistically select different AI-RAN operators, i.e., making the mixed selection strategy. Subsequently, we utilize the mixed matching profile to stabilize the AI-RAN operator-side contract update. Meanwhile, we introduce a shadow price to represent the AI user utility function that penalizes overloaded AI-RAN operators during the user-side selection step \cite{kelly1998rate}. The iteration of the fixed-point algorithm is running between the user-side and the AI-RAN operator side. Notably, after the iteration converges, we project the mixed matching profile to a deterministic matching outcome.

We can implement our proposed algorithm in a distributed manner. Specifically, each AI-RAN operator only needs to publicly announce contract menus and the congestion state. The internal optimization details and private cost parameters of AI-RAN operators remain hidden during the algorithm iteration.

\subsection{Equilibrium Contract Menu Profile Derivation} \label{sec:6.2}
Prior to the algorithm iteration running, we initialize the contract menu of each AI-RAN operator using the single-operator contract design in Section \ref{sec:5.2}. Subsequently, we iteratively update the mixed matching profile, shadow prices, and contract menus. We use $k$ to denote the iteration index.      

\subsubsection{Mixed User-Side Matching Update}
To ensure that the mixed matching profile is capacity-aware, we first define the effective service capacity of the AI-RAN operator $o_m$ as
\begin{equation} \label{eq:31}
    \bar{\lambda}_m \triangleq \chi \cdot \min \left\{ c_m^{\mathrm{UL}} \mu_m^{\mathrm{UL}}, c_m^{\mathrm{P}} \mu_m^{\mathrm{P}}, c_m^{\mathrm{DL}} \mu_m^{\mathrm{DL}} \right\},
\end{equation}
where $\chi \in (0,1]$ is a safety coefficient. Equation \eqref{eq:31} indicates that the effective service capacity of $o_m$ is determined by the bottleneck stage among uplink, processing, and downlink.

At iteration $k$, given the current contract menu profile $\mathcal{S}_c^{(k)}$ and the current congestion vector $\Lambda^{(k)}$, we define the shadow-price-adjusted utility of the type-$n$ AI user when selecting $o_m$ as
\begin{equation} \label{eq:32}
    \hat{\pi}_{u,m,n}^{(k)} =
    \pi_u(\mathcal{C}_m^{(k)};\theta_n,\mathcal{C}_{-m}^{(k)})
    - \omega_m^{(k)} \frac{|\mathcal{I}_n|\delta}{\bar{\lambda}_m},
\end{equation}
where $\omega_m^{(k)} \ge 0$ is a shadow price introduced coefficient, which is used to penalize overloaded AI-RAN operators during matching.

As per \eqref{eq:32}, we define the mixed matching probability of the type-$n$ AI user regarding the AI-RAN operator $o_m$ as
\begin{equation} \label{eq:33}
    z_{m,n}^{(k)} =
    \frac{
        \exp\left( \hat{\pi}_{u,m,n}^{(k-1)} / \tau^{(k)} \right)
    }{
        \exp\left( u_0 / \tau^{(k)} \right) +
        \sum_{m'=1}^{M}\exp\left( \hat{\pi}_{u,m',n}^{(k-1)} / \tau^{(k)} \right)
    },
\end{equation}
where $\tau^{(k)} > 0$ is the temperature parameter and $u_0$ denotes the opt-out utility. Notably, to ensure clarity, we use $\mathbf{Z}$ to denote the mixed AI user selection matrix, which is different from the discrete form AI user selection matrix $\mathbf{A}$. We use $z_{0,n}^{(k)}$ to denote the probability that the type-$n$ AI user opts out. As $\tau^{(k)}$ gradually decreases according to an annealing schedule \cite{kirkpatrick1983optimization}, the mixed matching probability in \eqref{eq:33} approaches a near-deterministic selection rule.

Upon the mixed matching matrix $\mathbf{Z}^{(k+1)}=\{Z_{m,n}^{(k)}\}$, we compute the cumulative traffic load of AI-RAN operator $o_m$ for the type-$n$ AI user as
\begin{equation} \label{eq:34}
    \lambda_{m,n}^{(k)} = \sum_{j=1}^{n} |\mathcal{I}_j| z_{m,j}^{(k)} \delta.
\end{equation}
To further stabilize the matching update, we use a damping step and define the final mixed matching profile of iteration $k$ as
\begin{equation} \label{eq:35}
    \mathbf{Z}^{(k)} \leftarrow (1-\vartheta)\mathbf{X}^{(k-1)} + \vartheta \hat{\mathbf{X}}^{(k)},
\end{equation}
where $\vartheta \in (0,1]$ is the matching damping coefficient. Subsequently, we update the shadow price coefficient of each AI-RAN operator as
\begin{equation} \label{eq:36}
    \omega_m^{(k+1)} =
    \left[
        \omega_m^{(k)}
        + \upsilon \frac{
            \sum_{n=1}^{N} |\mathcal{I}_n| z_{m,n}^{(k)} \delta - \bar{\lambda}_m
        }{\bar{\lambda}_m}
    \right]_+,
\end{equation}
where $\upsilon > 0$ is the shadow-price step size. Eq. \eqref{eq:36} implies that the shadow price will increase only when the aggregate demand allocated to $o_m$ exceeds its effective service capacity.

\begin{algorithm}[!t]
\DontPrintSemicolon
\SetKwInOut{Input}{Input}\SetKwInOut{Output}{Output}
\SetKwFunction{MixedResp}{MixedResp}
\SetKwFunction{Iron}{Iron}
\SetKwFunction{Recover}{Recover}

\Input{Initial values of $\vartheta$, $\upsilon$, $\tau^{(0)}$, and $\tau^{(K)}$; maximum iteration number $K$.}
\Output{Equilibrium contract menu profile $\mathcal{S}_c^\star$ and mixed matching profile $\mathbf{Z}^\star$.}

Initialize $\mathcal{S}_c^{(0)}$ via the no-competition contract design\;
Initialize $\mathbf{Z}^{(0)}$ and $\boldsymbol{\omega}^{(0)}$\;

\For{$k \leftarrow 1$ \KwTo $K$}{
    Derive the temperature $\tau^{(k)}$ according to the annealing schedule\;
    Derive the cumulative congestion vector $\Lambda^{(k)}$ via \eqref{eq:34} using $\mathbf{Z}^{(k-1)}$\;

    \ForEach{$m \in [M]$ \textbf{in parallel}}{
        Solve the relaxed problem of \eqref{eq:38} and obtain $\mathbf{L}_{m}^{(k),\mathrm{relax}}$\;
        $\mathbf{L}_{m}^{(k)} \leftarrow \Iron\!\left(\mathbf{L}_{m}^{(k),\mathrm{relax}}\right)$\;
        $\mathbf{R}_{m}^{(k)} \leftarrow \Recover\!\left(\mathbf{L}_{m}^{(k)}\right)$\;
        Update $\mathcal{C}_m^{(k)}=\{(L_{m,n}^{(k)},R_{m,n}^{(k)}) \mid n\in[N]\}$\;
    }

    Use \eqref{eq:32} and \eqref{eq:33} to derive the mixed response
    $\hat{\mathbf{Z}}^{(k)} \leftarrow \MixedResp(\mathcal{S}_c^{(k)},\Lambda^{(k)},\boldsymbol{\omega}^{(k)},\tau^{(k)})$\;

    Update the mixed matching profile via \eqref{eq:35}:
    $\mathbf{Z}^{(k)} \leftarrow (1-\vartheta)\mathbf{Z}^{(k-1)} + \vartheta \hat{\mathbf{Z}}^{(k)}$\;

    Update the shadow price vector via \eqref{eq:36}\;

    \If{convergence criterion is satisfied}{
        break\;
    }
}

Set $\mathbf{Z}^\star \leftarrow \mathbf{Z}^{(k)}$\;
Compute $\Lambda^\star$ from $\mathbf{Z}^\star$ via \eqref{eq:34}\;
For each $m\in[M]$, recompute the final contract menu $\mathcal{C}_m^\star$ under $\Lambda^\star$\;
Set $\mathcal{S}_c^\star=\{\mathcal{C}_m^\star \mid m\in[M]\}$\;
\Return{$\mathcal{S}_c^\star=\{\mathcal{C}_m^\star \mid m\in[M]\}$ and $\mathbf{Z}^\star$}\;

\caption{Mixed Stable Matching with Contracts}
\label{alg:1}
\end{algorithm}

\subsubsection{Operator-Side Contract Menu Update}
Upon the induced congestion vector $\Lambda^{(k)}=\{\lambda_{m,n}^{(k)}\}$, each AI-RAN operator independently updates its own contract menu. To align with the mixed matching profile, we define the effective demand mass of the type-$n$ AI user for AI-RAN operator $o_m$ as
\begin{equation} \label{eq:37}
    \hat{d}_{m,n}^{(k)} \triangleq |\mathcal{I}_n|\delta \left( (1-\rho)x_{m,n}^{(k)} + \rho \right),
\end{equation}
where $\rho \in (0,1)$ is a demand-floor coefficient. The parameter $\rho$ ensures that each contract item retains a small prior demand weight during the menu redesign step, which improves the numerical stability of the contract optimization.

Given $\hat{d}_{m,n}^{(k)}$ and $\Lambda^{(k)}$, AI-RAN operator $o_m$ solves
\begin{subequations} \label{eq:38}
    \begin{align}
        \max_{\mathbf{L}_m^{(k)}} \quad
        & \sum_{n=1}^{N} \hat{d}_{m,n}^{(k)}
        \Big(
        R_{m,n}(\mathbf{L}_m^{(k)})
        - \bar{C}\tilde{p}(L_{m,n}^{(k)};\lambda_{m,n}^{(k)})
        \Big) \label{eq:38a} \\
        \text{s.t.}\quad
        & L_{m,1}^{(k)} \le L_{m,2}^{(k)} \le \cdots \le L_{m,N}^{(k)}. \label{eq:38b}
    \end{align}
\end{subequations}

The problem \eqref{eq:38} is a single-dimensional contract screening problem. We solve it numerically via the following three steps:
\begin{enumerate}
    \item \textbf{Relaxed latency solve:} we first omit the monotonicity constraint \eqref{eq:38b} and solve a relaxed nonlinear but convex optimization problem to obtain $\mathbf{L}_{m}^{(k),\mathrm{relax}}$.
    \item \textbf{Iterative ironing via range re-solve:} if $\mathbf{L}_{m}^{(k),\mathrm{relax}}$ violates monotonicity, we identify the violating block and re-solve the optimization problem only on this block while fixing the non-violating entries. This procedure is repeated until a monotone latency vector is derived \cite{xiong2020multi}.
    \item \textbf{Reward recovery:} after obtaining the monotone latency vector $\mathbf{L}_{m}^{(k)}$, we recover the corresponding price vector $\mathbf{R}_{m}^{(k)}$ via Proposition~\ref{prop:2}.
\end{enumerate}
% Since the contract menu update of each AI-RAN operator only depends on its local congestion state and the publicly announced market outcome, all AI-RAN operators can perform this step in parallel.

\subsection{Algorithm Analysis} \label{sec:6.3}
We present the pseudocode of the proposed mixed stable matching with contracts algorithm in Algorithm~\ref{alg:1}. Before analyzing its computational complexity, we show that the mixed equilibrium point targeted by Algorithm~\ref{alg:1} exists.

\begin{theorem}[Existence of Mixed Equilibrium]
\label{thm:mixed_equilibrium}
Suppose that each latency agreement is selected from a bounded operational interval and the monotonicity constraint in \eqref{eq:38b} is imposed. Suppose further that the mixed matching profile is restricted to the queue-stable region induced by \eqref{eq:31}, and that the operator-side relaxed objective in \eqref{eq:38a} is quasi-concave in each operator's own latency vector. Then the game induced by Algorithm~\ref{alg:1} admits at least one mixed Nash equilibrium.
\end{theorem}

\begin{proof}
Consider the continuous game whose players are the AI-RAN operators and the AI user types. The strategy set of operator $o_m$ is
\[
\mathcal{L}_m=\{\mathbf{L}_m:\underline L\le L_{m,1}\le \cdots \le L_{m,N}\le \overline L\},
\]
which is nonempty, compact, and convex. Here, $(\underline{L}, \overline{L})$ is the optimization domain of the contract menu. The strategy set of type-$n$ AI users is the probability simplex.
\[
\Delta_n=\{\mathbf{z}_n\in\mathbb{R}_+^{M+1}:\sum_{m=0}^{M}z_{m,n}=1\},
\]
which is also nonempty, compact, and convex; the component $z_{0,n}$ represents opting out. Under queue stability, the Chernoff approximation $\tilde p(\cdot;\lambda_{m,n})$ in \eqref{eq:10} is continuous in the latency agreement and the induced load. Since the reward recovery in \eqref{eq:25}, the user utility in \eqref{eq:32}, and the operator utility in \eqref{eq:38a} are compositions of continuous functions, all players' payoffs are continuous. The user-side expected utility is linear, hence quasi-concave, in $\mathbf{z}_n$, and the operator-side quasi-concavity follows from the stated assumption on \eqref{eq:38a}. Therefore, the Debreu--Fan--Glicksberg existence theorem applies and guarantees a Nash equilibrium of this continuous game \cite{glicksberg1952generalization}.
\end{proof}

In the remainder of this subsection, we analyze the computational complexity of our proposed algorithm. We use $K$ to denote the maximum number of mixed iterations.

\textbf{Mixed user-side update.} At each iteration, every type-$n$ AI user evaluates its shadow-price-adjusted utility over all $M$ AI-RAN operators, which complexity is $\mathcal{O}(M)$. Therefore, deriving the mixed response of all AI user types requires $\mathcal{O}(MN)$ operations. The subsequent load update, damping update, and shadow-price update are also linear in the number of $(m,n)$ pairs. Hence, the overall complexity of the mixed user-side update at each iteration is $\mathcal{O}(MN)$.

\textbf{Operator-side contract menu update.} At each iteration, each AI-RAN operator solves one relaxed contract design problem over $N$ types, followed by iterative ironing and reward recovery. We use $\mathsf{C}_{\mathrm{relax}}(N)$ to denote the computational complexity of solving the relaxed nonlinear problem. The iterative ironing step repeatedly resolves violating blocks and, in the worst case, requires quadratic complexity in the number of AI user types. The reward recovery step is linear in $N$. Therefore, the wall-clock complexity of the operator-side update is
\begin{equation} \label{eq:39}
    \mathcal{O}\!\left(\mathsf{C}_{\mathrm{relax}}(N) + N^2\right),
\end{equation}
since all AI-RAN operators can execute their menu updates in parallel.

\textbf{Overall Complexity.} In a nutshell, the total wall-clock complexity of our proposed algorithm is
\begin{equation} \label{eq:40}
    \begin{aligned}
        \mathcal{O}\!\left(KMN\right)
        + & \mathcal{O}\!\left(K\left(\mathsf{C}_{\mathrm{relax}}(N)+N^2\right)\right).
    \end{aligned}
\end{equation}

\section{Experimental Analysis} \label{sec:7}
In this section, we will evaluate our proposed method by varying several key parameters. Firstly, we introduce the experimental configurations, including parameter settings, benchmarks, and evaluation metrics, in Section \ref{sec:7.1}. Then, we do the comparison analysis against three benchmarks, including contract menu analysis and parameters analysis, in Section \ref{sec:7.2}.

\subsection{Configurations} \label{sec:7.1}
We categorize the parameters in this paper into three parts: the AI task model, the AI-RAN market model, and the proposed algorithm. First, for the AI task model, we refer to our previous works that utilize a diffusion-based AIGC model in a Unity-based teleoperation project \cite{zhan2025distributionally, zhan2025vision}. We set the average input size, computation burden, and output size of the AI task as $d_i=0.18$ Mb, $\tau=3.6\times 10^{11}$ FLOPs, and $d_o=0.27$ Mb, respectively. We set the AI task arrival rate of each AI user as $\delta=24$ per second. Second, for the AI-RAN market model, we consider $M=3$ heterogeneous AI-RAN operators and $N=8$ AI user types unless otherwise specified. We set the total number of AI users as $\sum_u=150$ by default. The latency-sensitivity vector is $\boldsymbol{\beta}=[8,7,6,5,4,3,2,1]\times 10^{-4}$, and the user composition is generated by a Dirichlet distribution with a default setting $\alpha=10$. The common service quality factor is $q_m=1.5$, which is calculated via metrics of PSNR (Peak Signal-to-Noise Ratio) and SSIM (Structural Similarity Index Measure). By referring to \cite{zhan2025learning}, we set the refund to AI users as $\bar{R}=1.2\times 10^{-4}$, the operator-side latency violation cost as $\bar{C}=1.2\times 10^{-3}$, and the model execution cost as $8\times 10^{-6}$. We assume three AI-RAN operators have heterogeneous radio and computing resources, and their effective AI service capacities are $(c_m^{UL},c_m^{P},c_m^{DL})=(48,24,194)$, $(43,16,172)$, and $(28,12,115)$, respectively. Remarkably, we refer to \cite{etsi2020nr38211, etsi2025nr38104} to set the heterogeneous settings of AI-RAN operators. Finally, for the proposed algorithm, the Chernoff parameter is set as $\zeta=0.9$, the demand-floor ratio is set as $0.05$, the matching damping coefficient is set as $0.35$, and the maximum number of mixed iterations is set as $50$.

\begin{figure*}[!t]
    \centering
    \begin{subfigure}{0.32\linewidth}
        \centering
        \includegraphics[width=\linewidth]{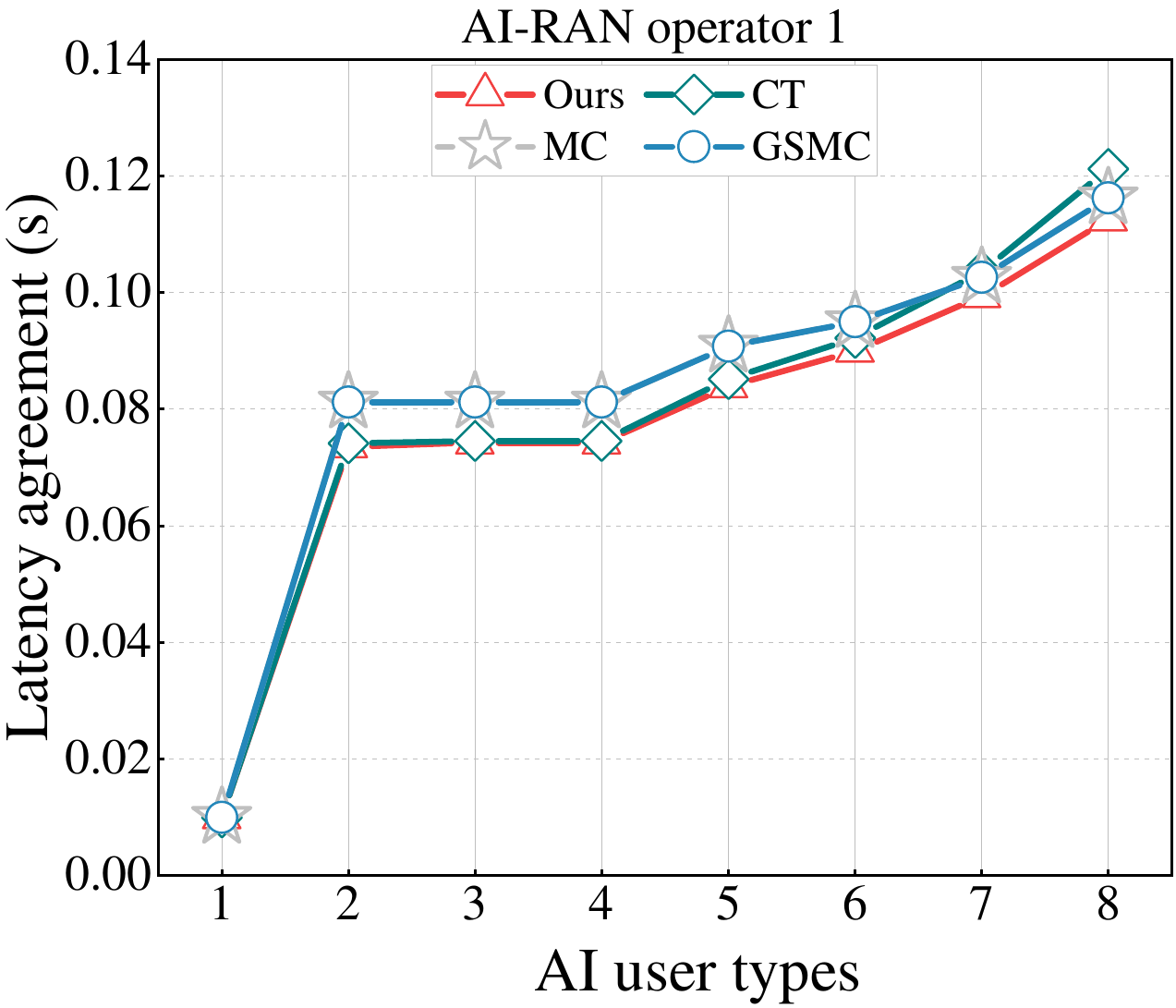}
        \caption{Latency agreements of the AI-RAN operator 1.}
        \label{fig3a}
    \end{subfigure}
    \hfill
    \begin{subfigure}{0.32\linewidth}
        \centering
        \includegraphics[width=\linewidth]{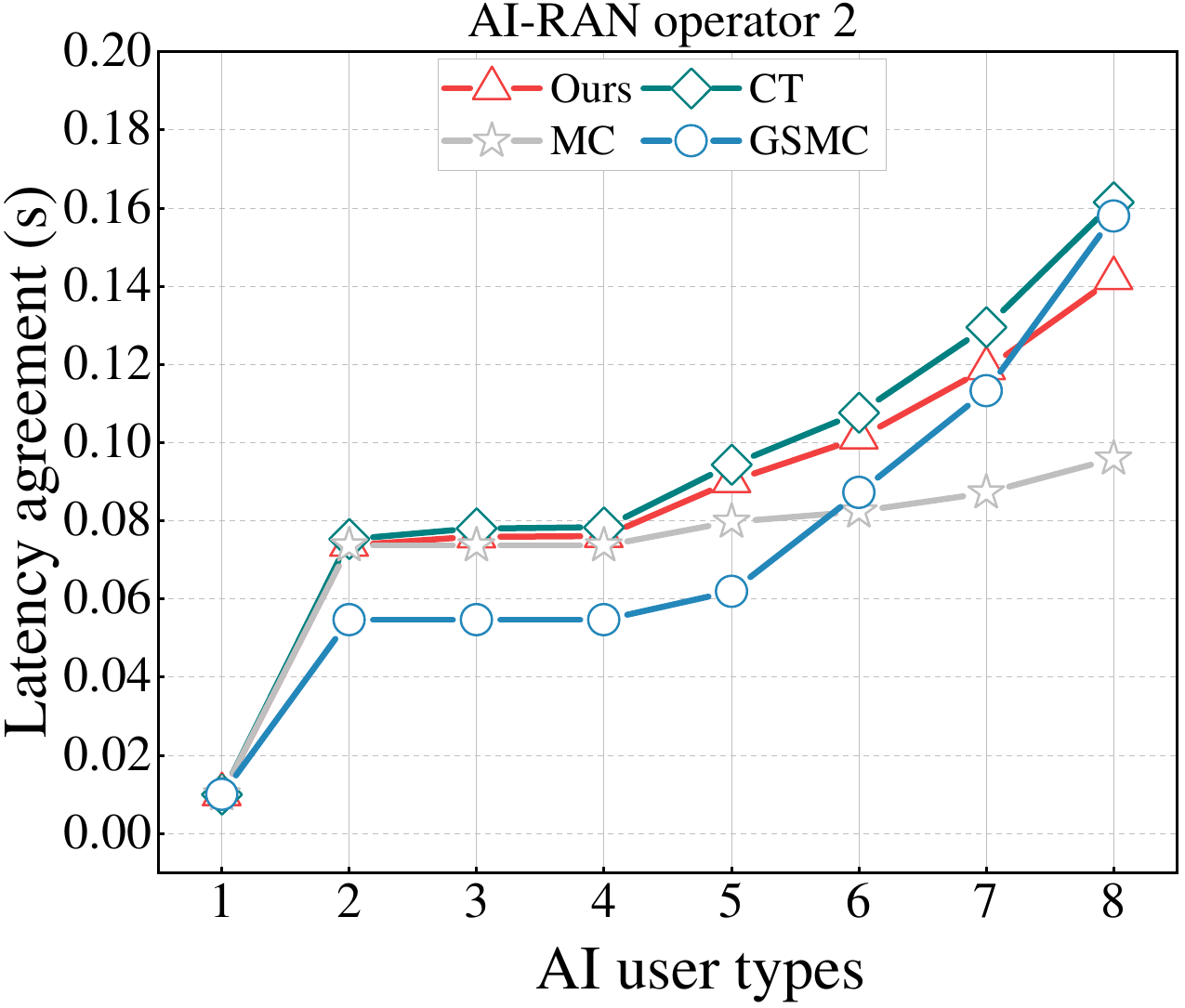}
        \caption{Latency agreements of the AI-RAN operator 2.}
        \label{fig3b}
    \end{subfigure}
    \hfill
    \begin{subfigure}{0.32\linewidth}
        \centering
        \includegraphics[width=\linewidth]{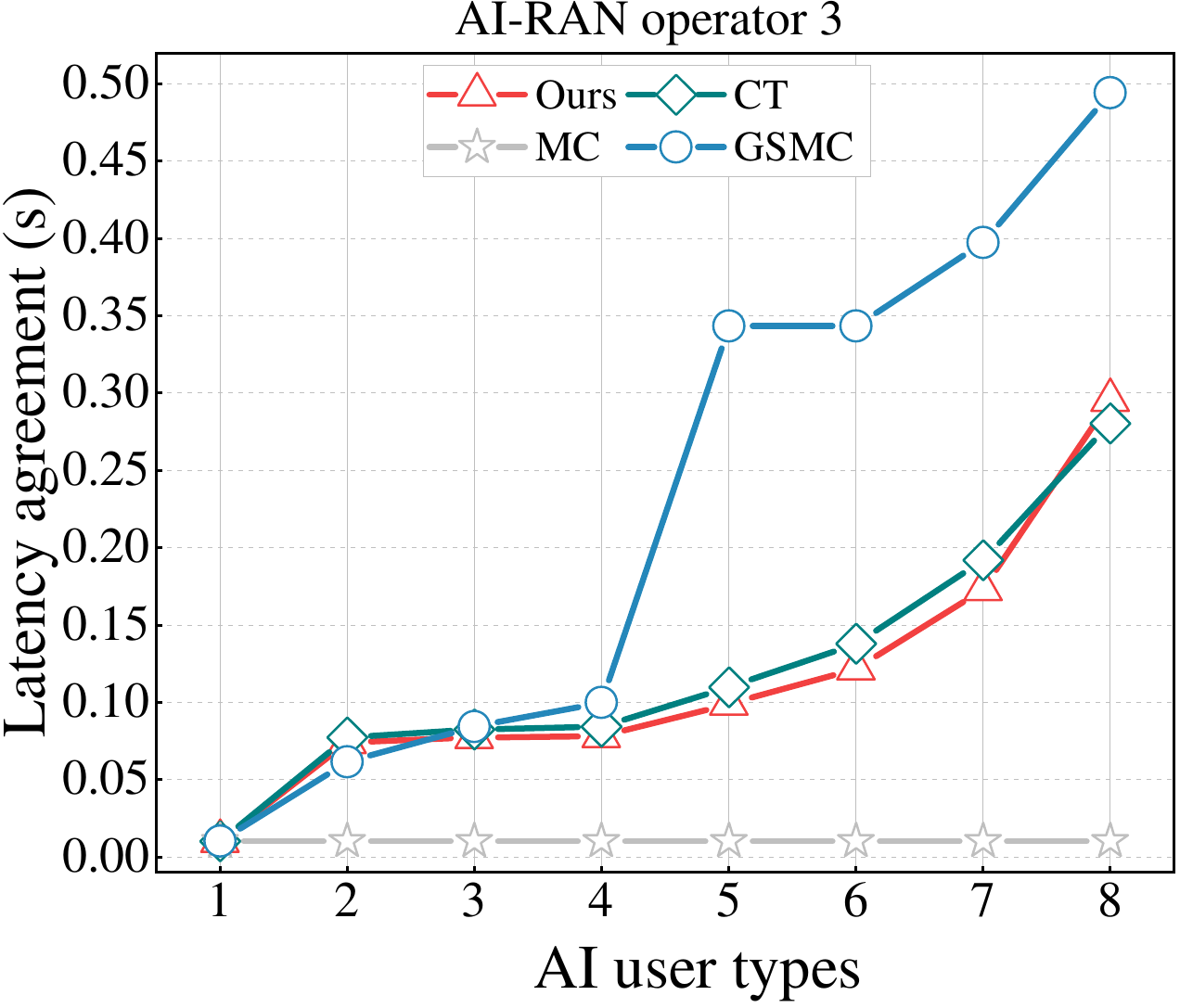}
        \caption{Latency agreements of the AI-RAN operator 3.}
        \label{fig3c}
    \end{subfigure}
    \caption{Equilibrium contract menus (latency agreements) of the three AI-RAN operators under the baseline setting.}
    \label{fig:3}
\end{figure*}

To assess the effectiveness of our proposed method, we compare it with the following benchmarks:
\begin{enumerate}
    \item [1)] Traditional Contract Theory (\textbf{CT}): Each AI-RAN operator independently derives its posted contract menu under a static default workload assumption, and AI users reselect AI-RAN operators under the posted contract menus.

    \item [2)] Static Matching-with-Contracts (\textbf{MC}): AI users first select operators under the original posted menus. Then, each AI-RAN operator redesigns its contract menu once under the induced matched workload.

    \item [3)] Gale-Shapely-based Matching-with-Contracts (\textbf{GSMC}): Gale-Shapley matching is conducted first, and the contract menu is redesigned after matching.
\end{enumerate}
To align with the baselines, we designate our proposed method as \textbf{Ours}. We use Python to conduct all simulations. We adopt two metrics to assess the performance of our proposed method, which are the total AI-RAN operator utility and the social welfare. The social welfare is defined as the sum of the total AI-RAN operator utility and the total user utility.
% Based on these metrics, we further study the effects of the market size $\sum_u\in\{30,60,90,120,150\}$, the number of AI user types $N\in\{4,6,8,10,12\}$, the refund coefficient $\bar{R}$, the operator-side violation cost $\bar{C}$, the Dirichlet parameter $\alpha\in\{0.1,1,10,100\}$, and the Chernoff parameter $\zeta\in\{0.7,0.75,0.8,0.85,0.9\}$.

\subsection{Experimental Results Analysis} \label{sec:7.2}
\subsubsection{Contract Menu Analysis}
We commence with analyzing the equilibrium contract menus under the default setting, and the results are depicted in Figs. \ref{fig3a}--\ref{fig3c}. Moreover, we present the matching results between AI-RAN operators and AI users under the default setting in Table \ref{table1}. Observing Fig. \ref{fig:3}, the contract menus of three AI-RAN operators derived by our proposed method and three benchmarks all satisfy the monotonicity constraint. Therefore, our proposed method and three benchmarks are feasible for deriving the contract menu in a competing AI-RAN service provision market. Moreover, since the AI-RAN operator 1 have the most available resource and the AI-RAN operator 3 has the least available resource, the contract menu derived by the AI-RAN operator 1 is most suitable for latency sensitive AI users.

\begin{table}[!t]
\caption{Matching results between AI user types and AI-RAN operators. Our proposed method is a mixed strategy, in which (0.25, 0.25, 0.25) means type 1 AI users have the same probability of 25\% select all of the AI-RAN operators. Notably, type 1 AI users also have 25\% probability select opt-out since all of three AI-RAN operators will extract all of the utility from type 1 AI users.}
\begin{center}
    \small
    \begin{tabular}{c|ccccc}
    \hline
    Type & Ours & CT & MC & GSMC \\ \hline
    1 & (0.25, 0.25, 0.25) & 1 & 1 & 1 \\
    2 & (0.256, 0.256, 0.256) & 1 & 1 & 2 \\
    3 & (0.29, 0.29, 0.29) & 3 & 1 & 1 \\
    4 & (0.306, 0.312, 0.317) & 3 & 1 & 1 \\
    5 & (0.315. 0.323. 0.335) & 3 & 1 & 3 \\
    6 & (0.294. 0.325. 0.368) & 3 & 1 & 3 \\
    7 & (0.252, 0.311, 0.432) & 3 & 1 & 3 \\
    8 & (0.164, 0.247, 0.589) & 3 & 1 & 3 \\ \hline
    \end{tabular}
\end{center}
\label{table1}
\vspace{-6mm}
\end{table}

By jointly observing Fig. \ref{fig:3} and Table \ref{table1}, we can analyze the rationale of the contract menu design of three AI-RAN operators. For our proposed method, the fist five AI user types almost have the same probability select three AI-RAN operators, which is because the contract menu provided by the three AI-RAN operators is similar. For the rest three AI user types, they have a higher probability; type 8 AI users have a probability of 58.9\%, to select the AI-RAN operator 3 since they are not sensitive to the latency. Analogously, for three benchmarks, the contract menu design is aligned with the mapping between AI user types and AI-RAN operators. For instance, regarding the benchmark MC, since all of the AI users are matched with AI-RAN operator 1, the contract menus derived by AI-RAN operators 2 and 3 do not vary across the eight AI user types.

\begin{figure}[!t]
    \centering
    \begin{subfigure}{0.49\linewidth}
        \centering
        \includegraphics[width=\linewidth]{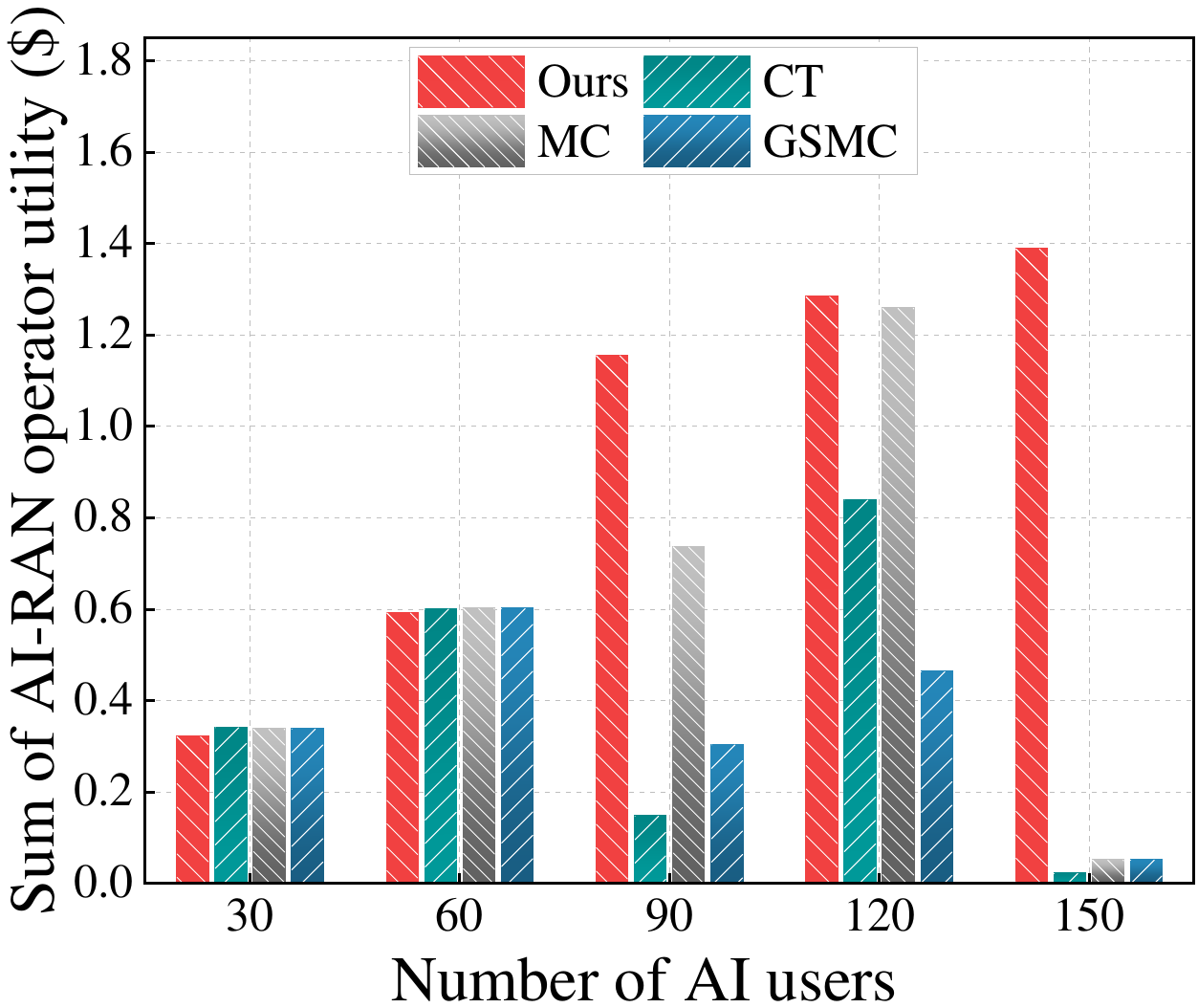}
        \caption{Total operator utility versus market size.}
        \label{fig4a}
    \end{subfigure}
    \hfill
    \begin{subfigure}{0.49\linewidth}
        \centering
        \includegraphics[width=\linewidth]{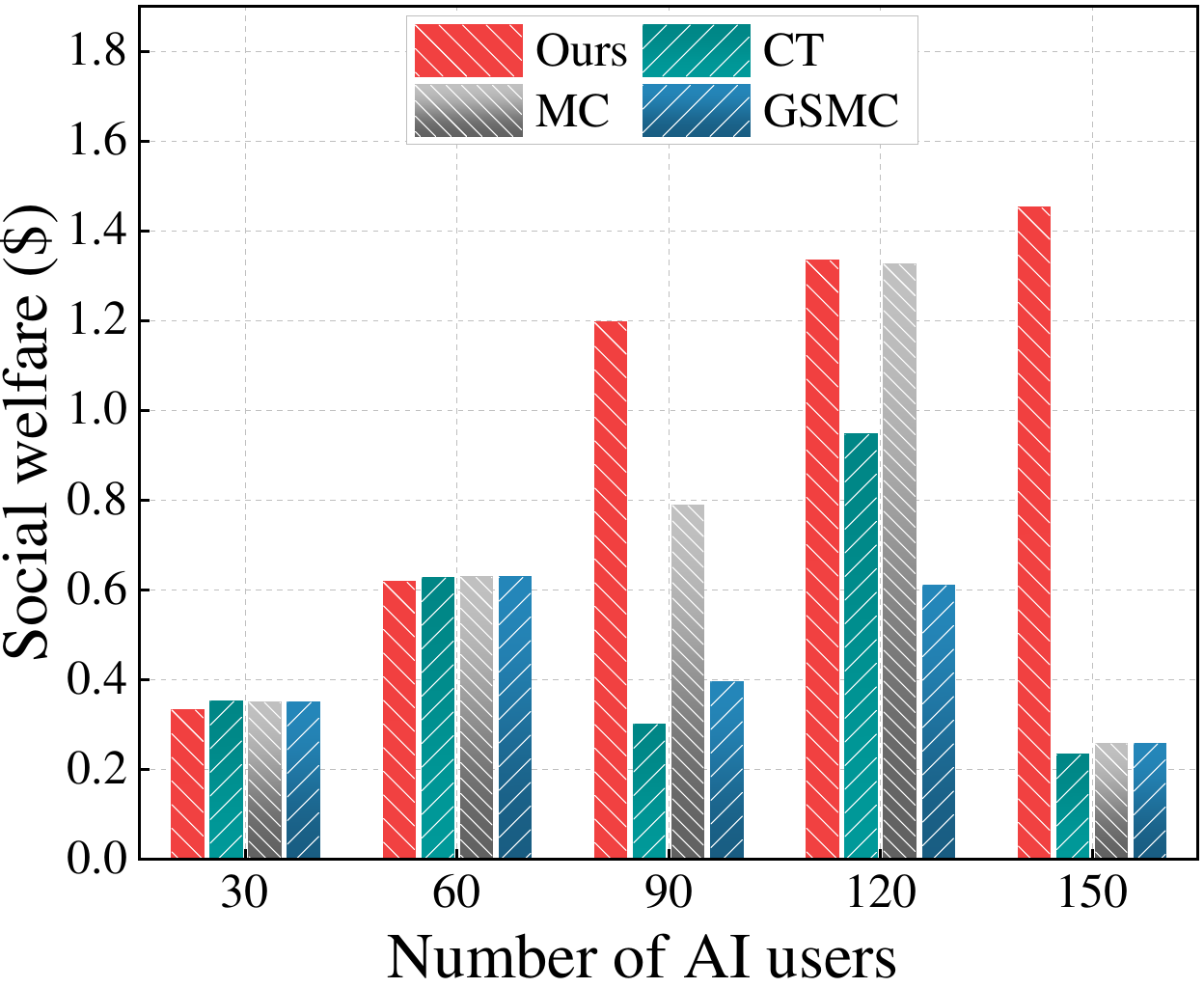}
        \caption{Social welfare versus market size.}
        \label{fig4b}
    \end{subfigure}

    \vspace{0.6em}

    \begin{subfigure}{0.49\linewidth}
        \centering
        \includegraphics[width=\linewidth]{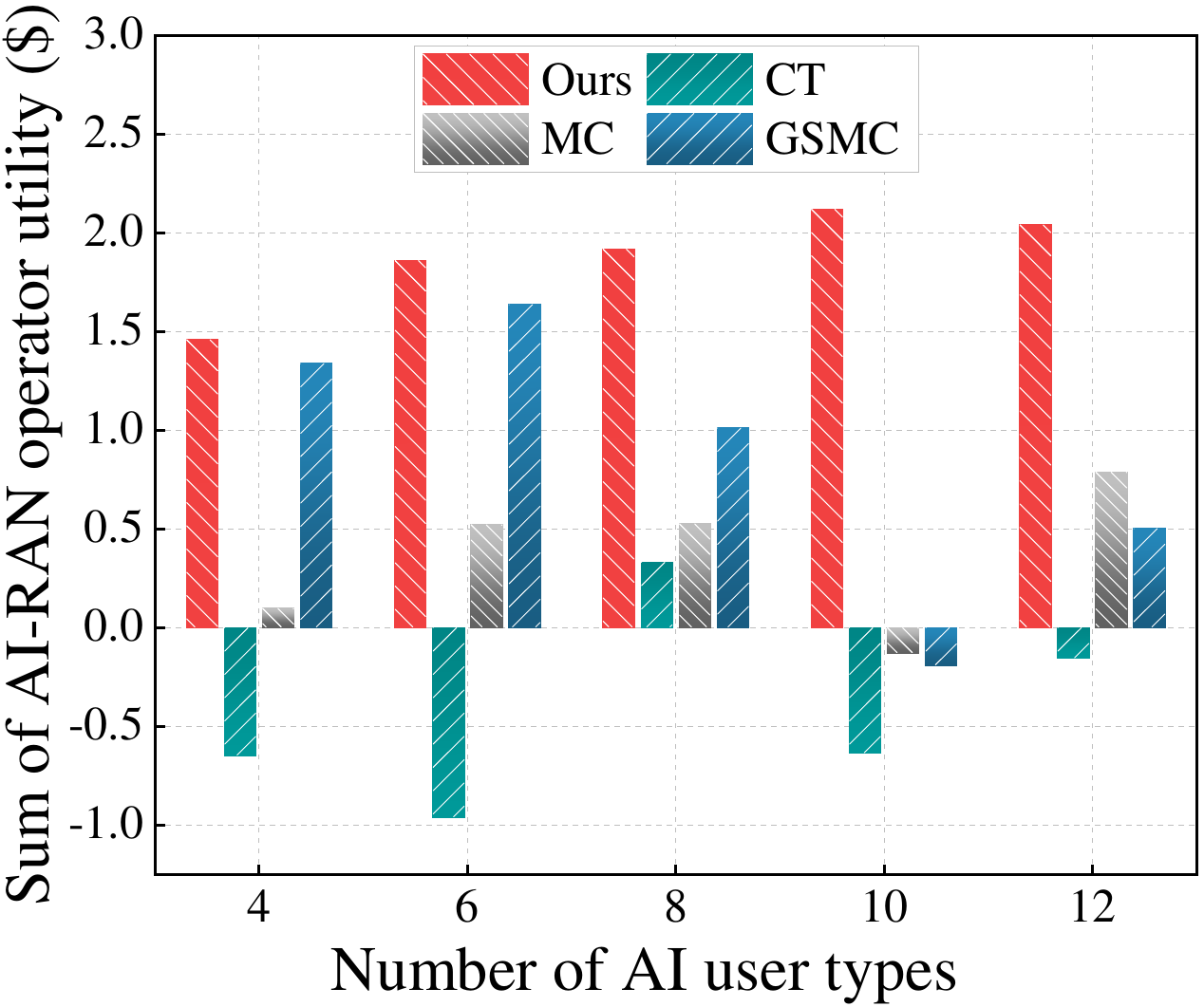}
        \caption{Total operator utility versus number of user types.}
        \label{fig4c}
    \end{subfigure}
    \hfill
    \begin{subfigure}{0.49\linewidth}
        \centering
        \includegraphics[width=\linewidth]{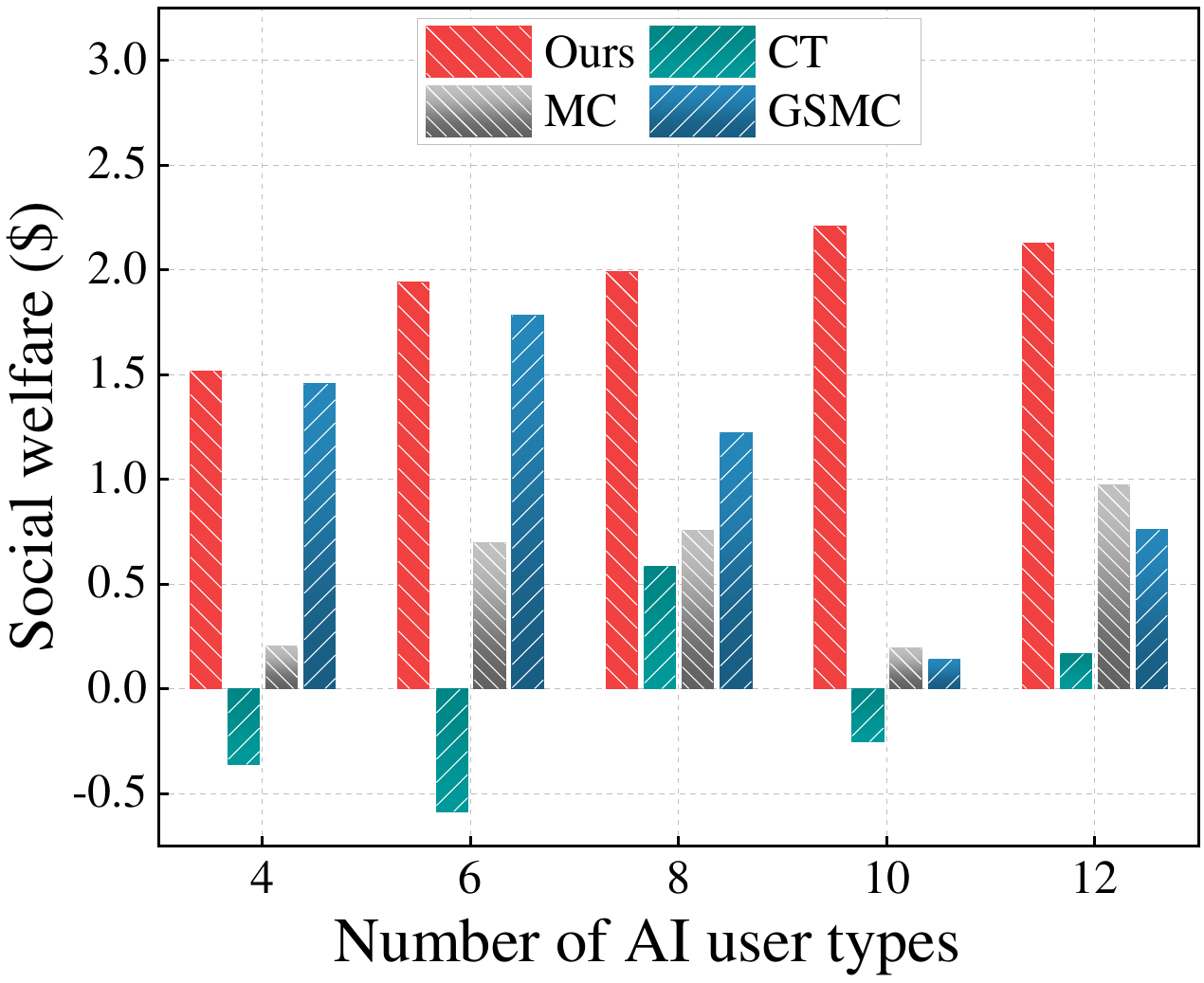}
        \caption{Social welfare versus number of user types.}
        \label{fig4d}
    \end{subfigure}
    \caption{Performance comparison under varying market size and varying number of AI user types.}
    \label{fig:4}
\end{figure}

\subsubsection{Impact of Market Size and User-Type Granularity}
In this section, we first study the impact of market size, and the results are presented in Figs. \ref{fig4a} and \ref{fig4b}. Figs.~\ref{fig4a} and \ref{fig4b} show the total AI-RAN operator utility and social welfare under different numbers of AI users. When the market size is small, congestion is weak, and one-shot benchmarks can still provide competitive outcomes. For example, when $\sum_u=30$, the proposed method achieves a total operator utility of about $0.322$ and a social welfare of about $0.334$, while the best benchmark, CT, achieves about $0.342$ and $0.354$, respectively. In this lightly loaded case, the proposed method is about $5.7\%$ lower in total operator utility and about $5.6\%$ lower in social welfare than CT. A similar pattern appears when $\sum_u=60$, where the best benchmark is only slightly higher than the proposed method.

The trend changes when the market becomes larger. As $\sum_u$ increases, congestion coupling becomes stronger, and the advantage of joint matching and contract redesign becomes clear. When $\sum_u = 90$, the proposed method improves the total operator utility and social welfare over the best benchmark by at least $56.8\%$ and $51.7\%$, respectively. When $\sum_u=120$, MC is close to the proposed method, but the proposed method still improves the total operator utility and social welfare by about $2.0\%$ and $0.6\%$, respectively. This result shows that the proposed method is especially useful when the induced congestion has a strong impact on service reliability.

Figs.~\ref{fig4c} and \ref{fig4d} show the impact of the number of AI user types. As $N$ increases, the market becomes more heterogeneous, and the screening problem becomes harder. For representative settings $N=8$ and $N=12$, the proposed method improves the total operator utility over the best benchmark by at least $89.6\%$, and improves the social welfare by at least $63.1\%$. When $N=10$, the best benchmark has a negative total operator utility of about $-0.129$, while the proposed method achieves a positive total operator utility of about $2.121$. Therefore, the proposed method provides a better balance between operator profit and user benefit when user heterogeneity becomes stronger.

\begin{figure}[!t]
    \centering
    \begin{subfigure}{0.49\linewidth}
        \centering
        \includegraphics[width=\linewidth]{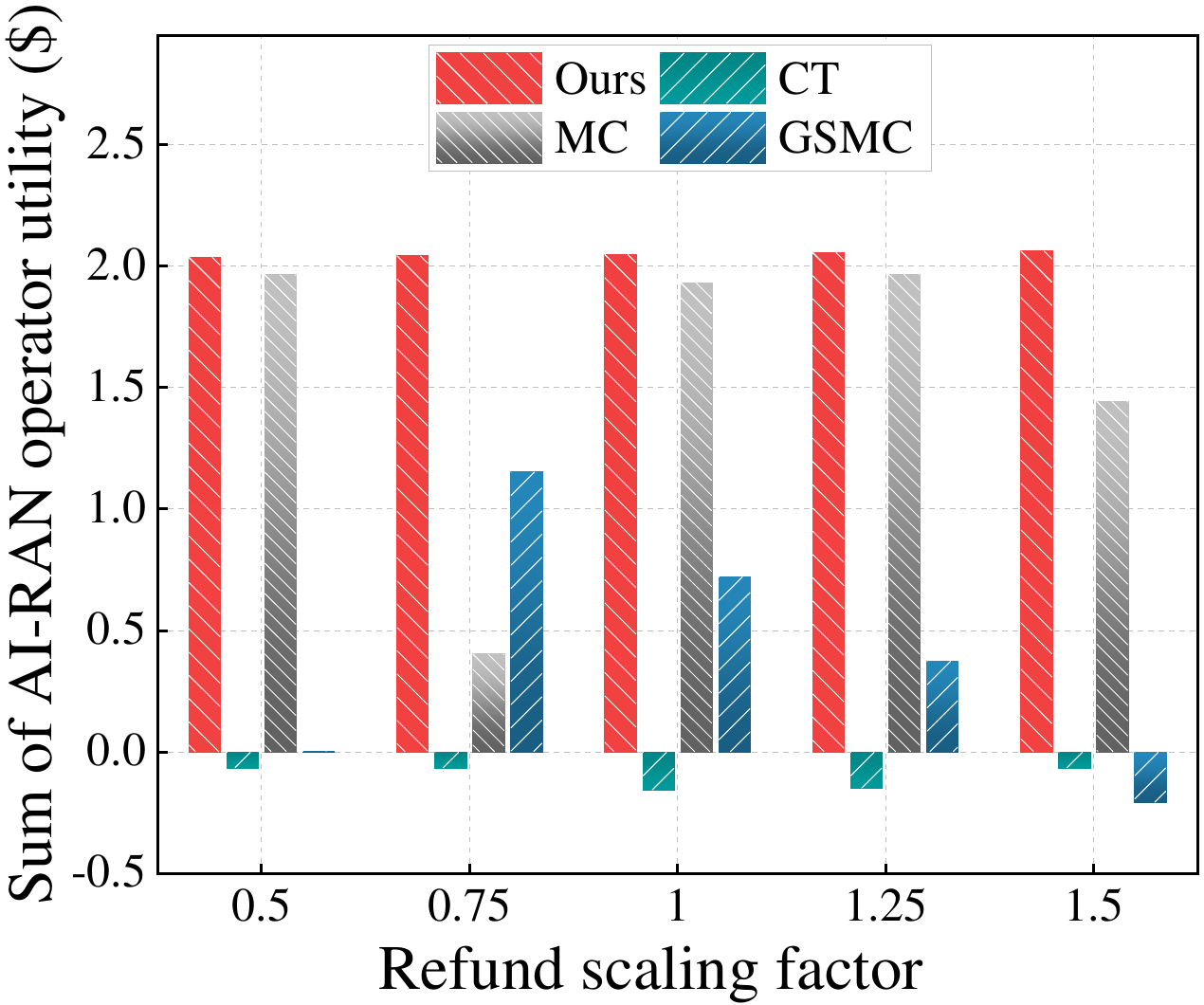}
        \caption{Total operator utility versus refund scaling factor.}
        \label{fig5a}
    \end{subfigure}
    \hfill
    \begin{subfigure}{0.49\linewidth}
        \centering
        \includegraphics[width=\linewidth]{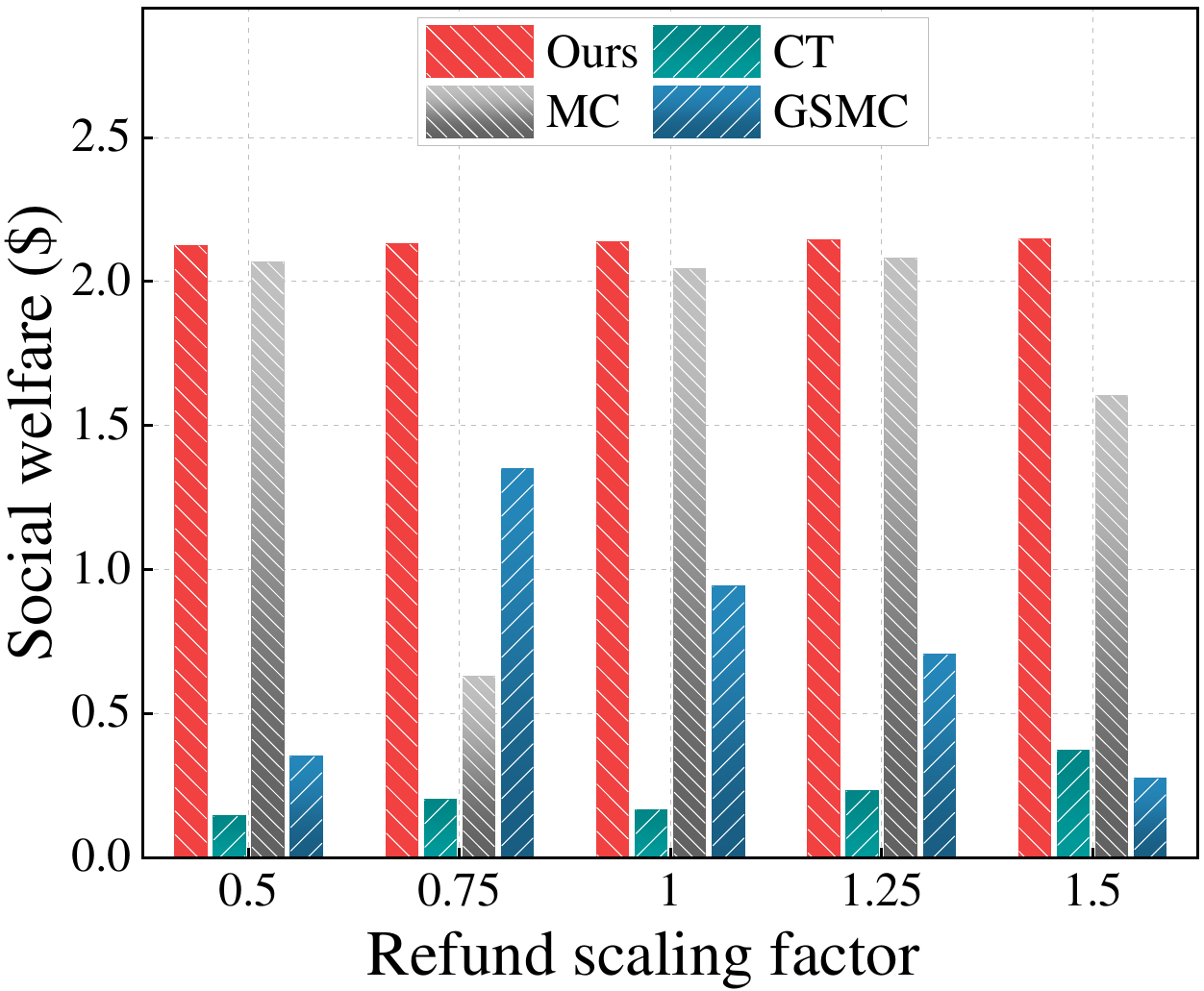}
        \caption{Social welfare versus refund scaling factor.}
        \label{fig5b}
    \end{subfigure}

    \vspace{0.6em}

    \begin{subfigure}{0.49\linewidth}
        \centering
        \includegraphics[width=\linewidth]{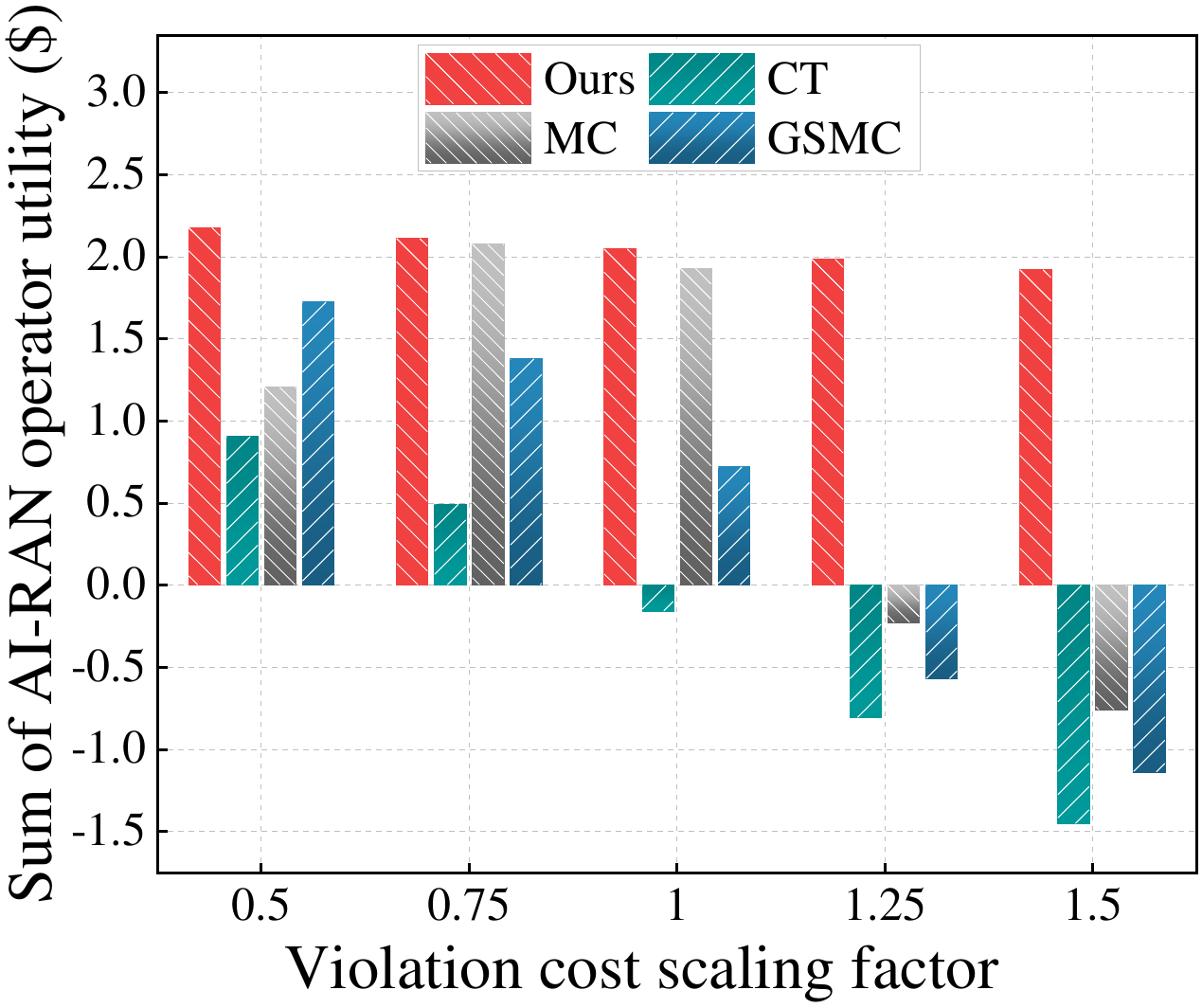}
        \caption{Total operator utility versus violation cost scaling factor.}
        \label{fig5c}
    \end{subfigure}
    \hfill
    \begin{subfigure}{0.49\linewidth}
        \centering
        \includegraphics[width=\linewidth]{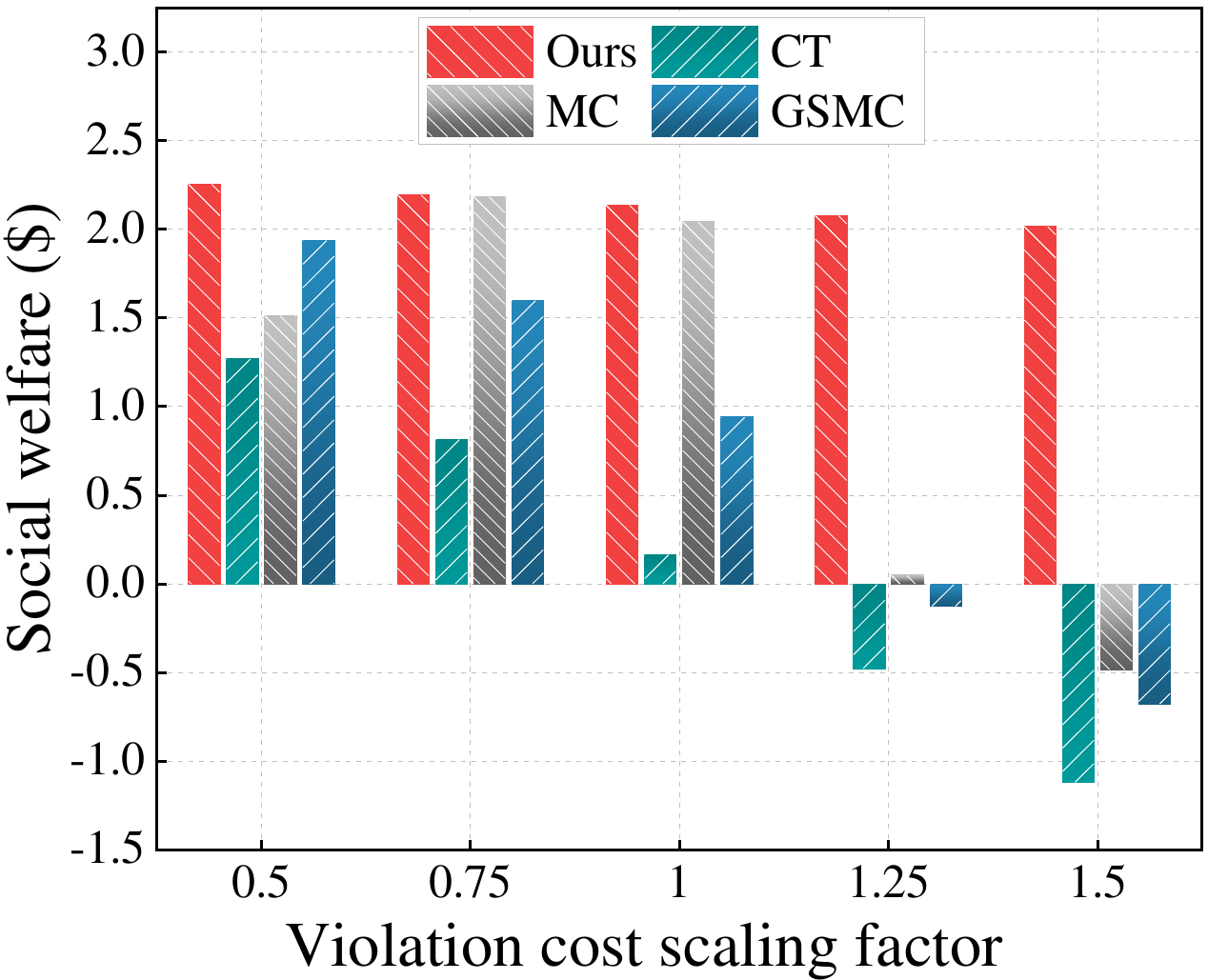}
        \caption{Social welfare versus violation cost scaling factor.}
        \label{fig5d}
    \end{subfigure}
    \caption{Performance comparison under varying refund and operator-side violation cost.}
    \label{fig:5}
\end{figure}

\subsubsection{Impact of Economic Parameters}
We next investigate two economic parameters, the refund coefficient $\bar{R}$ and the operator-side violation cost $\bar{C}$. Figs.~\ref{fig5a} and \ref{fig5b} show the results under different refund levels. A larger refund provides stronger compensation to users when the latency agreement is violated. At the same time, it also changes the operators' incentives in contract design.

Across all tested refund levels, the proposed method achieves the highest total operator utility and social welfare. When the refund scaling factor increases from $0.5$ to $1.5$, the total operator utility of the proposed method remains around $2.036$--$2.062$, and its social welfare increases slightly from about $2.126$ to $2.150$. For representative refund scaling factors $0.5$, $1$, and $1.5$, the proposed method improves the total operator utility over the best benchmark by about $3.7\%\sim42.9\%$, and improves the social welfare by about $2.8\%\sim33.9\%$. These results show that the proposed method is robust against the variation of the refund value.

Figs.~\ref{fig5c} and \ref{fig5d} show the results under different violation costs. As the violation cost increases, operators become more sensitive to latency-agreement violation, and the market becomes more risk-sensitive. For representative violation-cost scaling factors $0.5$ and $1$, the proposed method improves the total operator utility over the best benchmark by about $6.1\%\sim26.2\%$, and improves the social welfare by about $4.6\%\sim16.4\%$. When the scaling factor increases to $1.5$, all three benchmarks produce negative social welfare, while the proposed method still achieves a positive total operator utility of about $1.922$ and a positive social welfare of about $2.017$. These results verify the benefit of jointly updating the matching outcome and the contract menus under stricter reliability penalties.

\begin{figure}[!t]
    \centering
    \begin{subfigure}{0.49\linewidth}
        \centering
        \includegraphics[width=\linewidth]{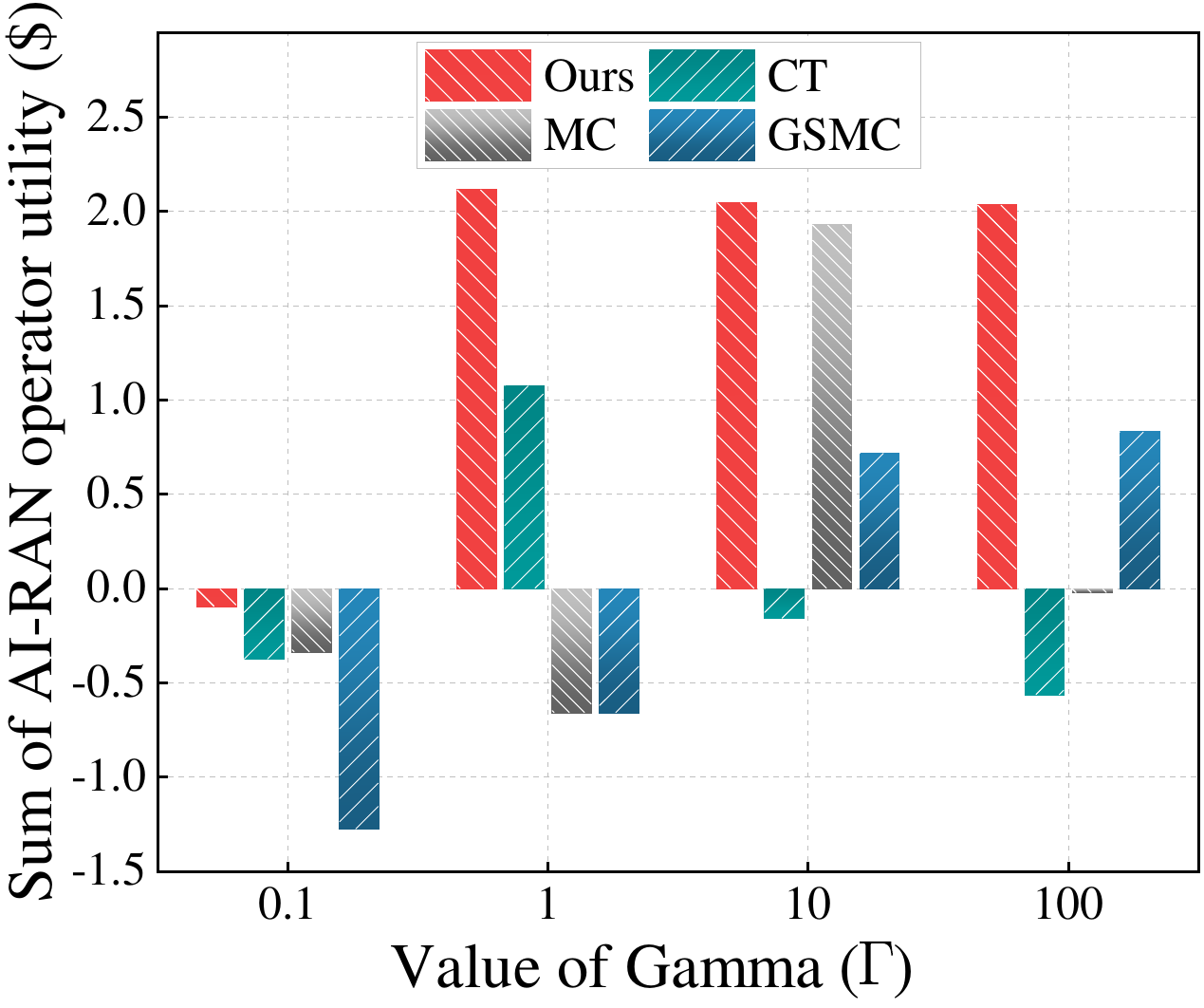}
        \caption{Total operator utility versus Dirichlet parameter $\alpha$.}
        \label{fig6a}
    \end{subfigure}
    \hfill
    \begin{subfigure}{0.49\linewidth}
        \centering
        \includegraphics[width=\linewidth]{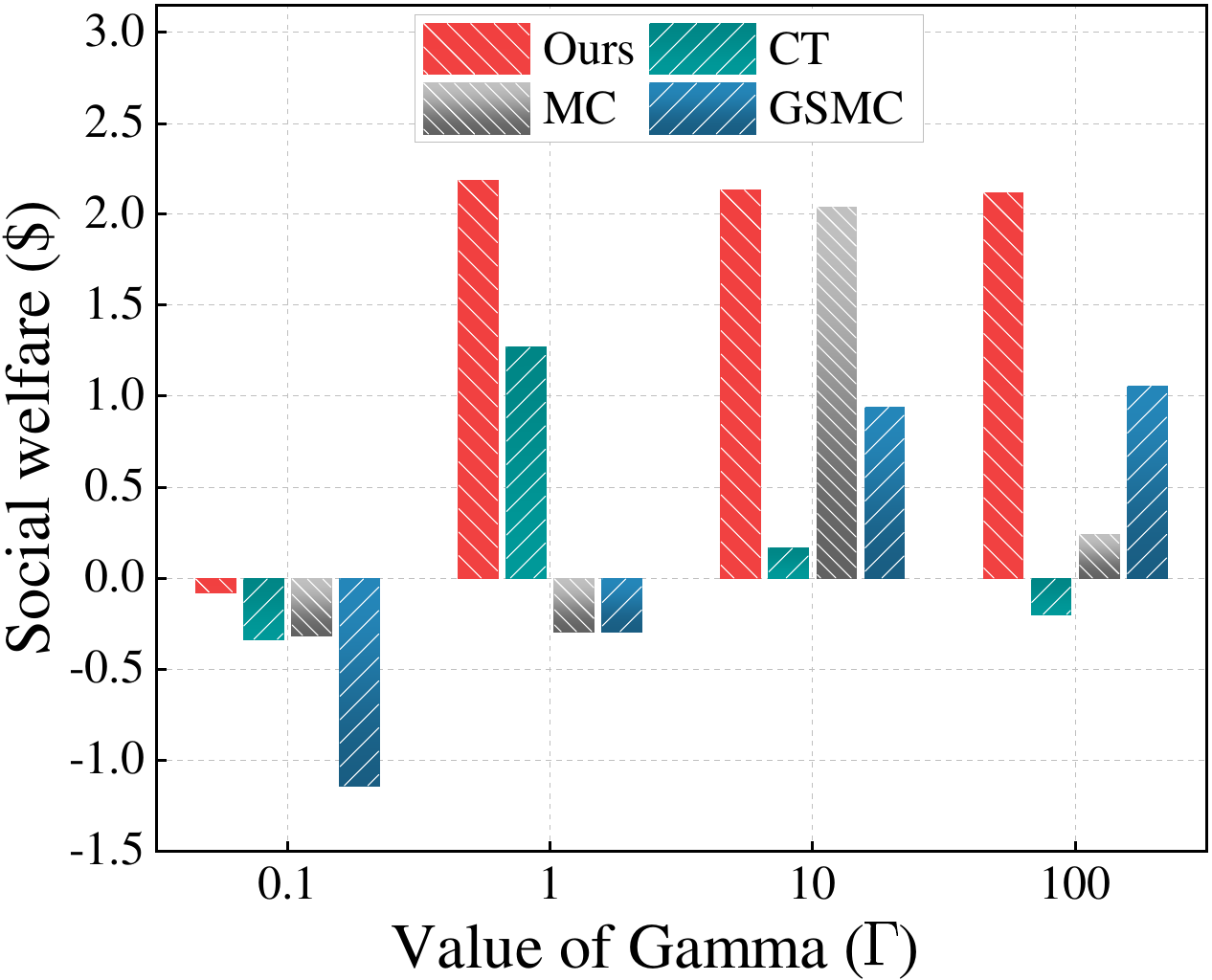}
        \caption{Social welfare versus Dirichlet parameter $\alpha$.}
        \label{fig6b}
    \end{subfigure}

    \vspace{0.6em}

    \begin{subfigure}{0.49\linewidth}
        \centering
        \includegraphics[width=\linewidth]{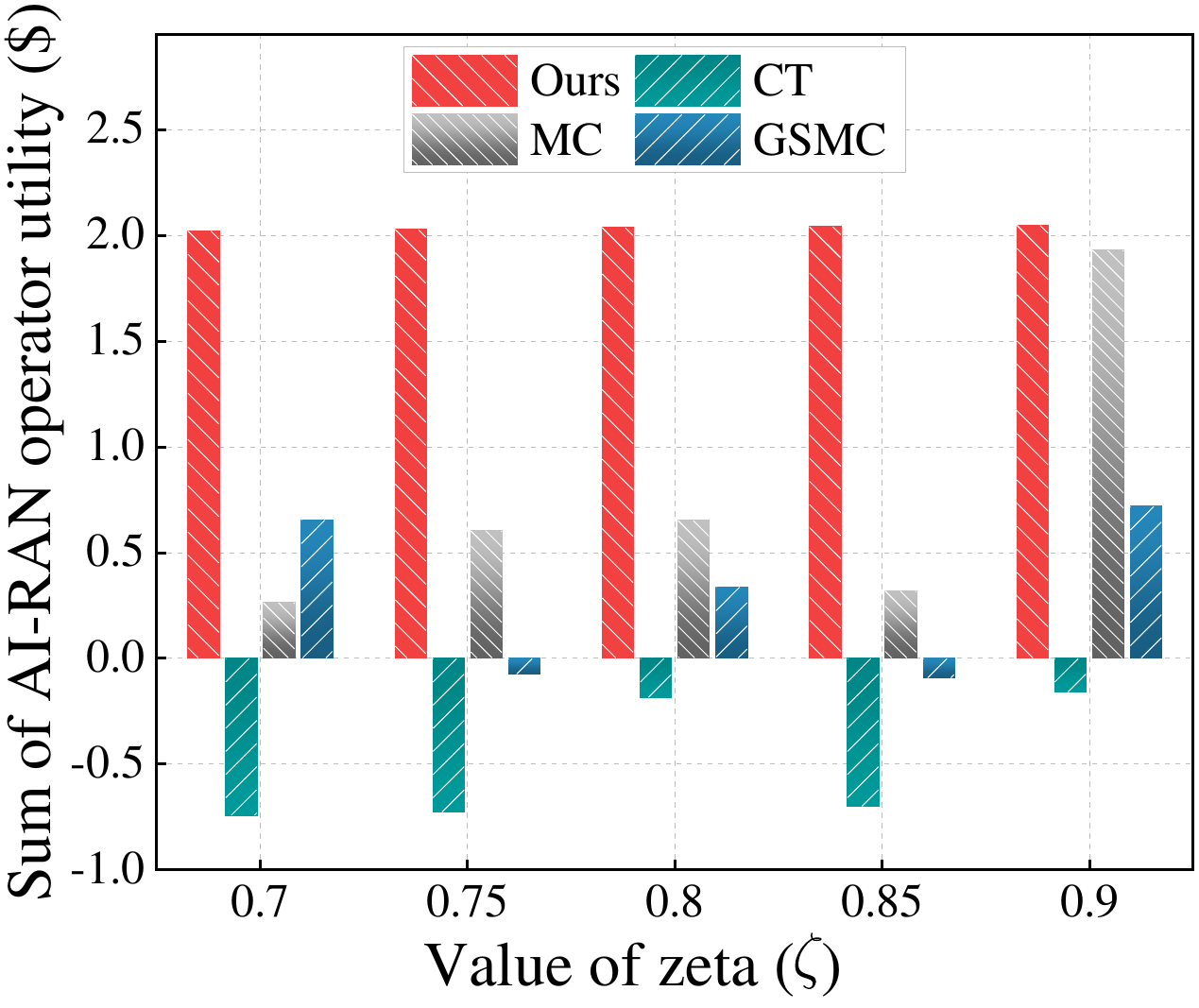}
        \caption{Total operator utility versus Chernoff parameter $\zeta$.}
        \label{fig6c}
    \end{subfigure}
    \hfill
    \begin{subfigure}{0.49\linewidth}
        \centering
        \includegraphics[width=\linewidth]{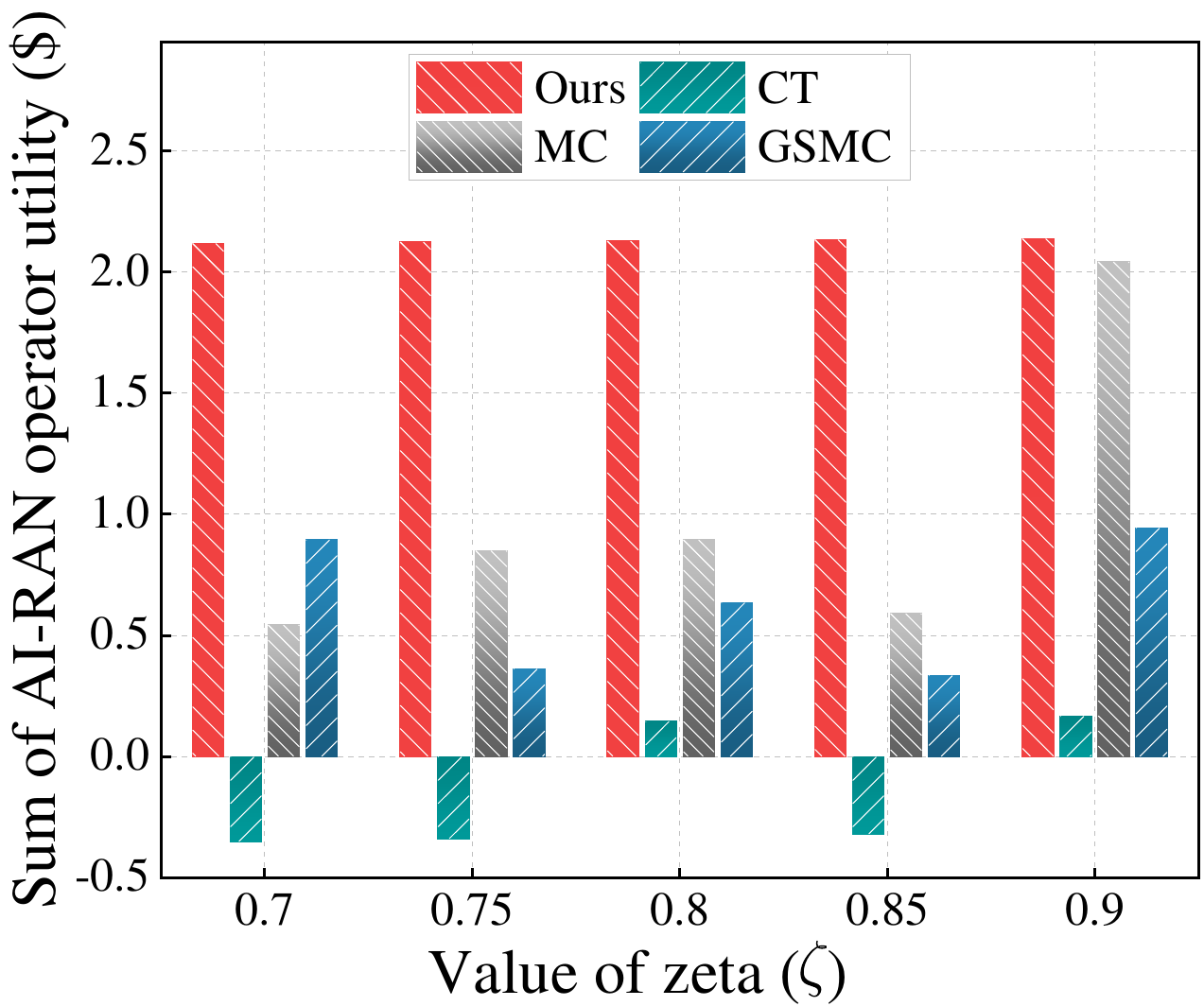}
        \caption{Social welfare versus Chernoff parameter $\zeta$.}
        \label{fig6d}
    \end{subfigure}
    \caption{Performance comparison under varying user composition and Chernoff parameter.}
    \label{fig:6}
\end{figure}

\subsubsection{Impact of User Type Distribution and Chernoff Parameter}
Finally, we evaluate the impact of the AI user types distribution and the Chernoff parameter. Figs.~\ref{fig6a} and \ref{fig6b} showcase the results under different Dirichlet parameters $\alpha$. This parameter controls the skewness of the AI user distribution. When $\alpha=0.1$, the market is highly imbalanced, and all methods yield negative total operator utility and negative social welfare. Even in this difficult case, the proposed method is the least negative one. Its social welfare is about $-0.076$, compared with $-0.336$, $-0.316$, and $-1.143$ for CT, MC, and GSMC, respectively.

When $\alpha$ becomes larger, the AI user types become more balanced, and the market performance improves significantly. For representative settings $\alpha=1$ and $\alpha=100$, the proposed method improves the total operator utility over the best benchmark by at least $96.4\%$, and improves the social welfare by at least $72.1\%$. This result indicates that the proposed framework is robust to different user composition patterns.

Figs.~\ref{fig6c} and \ref{fig6d} show the results under different Chernoff parameters $\zeta$. The proposed method remains stable across the tested range and achieves the highest total operator utility and social welfare for all tested values. For representative settings $\zeta=0.7$ and $\zeta=0.85$, the proposed method improves the total operator utility over the best benchmark by at least $208.7\%$, and improves the social welfare by at least $136.8\%$. Therefore, the proposed method is robust to different approximation settings in the latency violation model.

\section{Conclusion} \label{sec:8}
In this paper, we studied incentive mechanism design for AI task offloading in a competitive AI-on-RAN service market. We formulated a latency-price contract design problem under information asymmetry and competition, and proposed a mixed stable matching with contracts algorithm that jointly updates user-side matching and operator-side contract menus. Notably, we extend the conventional static matching-with-contracts model by jointly considering contract menus design of multiple competitive principals, principals-agents matching, and dynamic market-state evolution. Furthermore, we prove the existence of a mixed Nash equilibrium of our proposed dynamic matching-with-contracts formulation. Numerical results showed that the proposed method can derive monotone and interpretable contract menus for heterogeneous AI-RAN operators. Remarkably, the proposed method consistently improves market outcomes over benchmarks CT, MC, and GSMC when congestion, user heterogeneity, and reliability penalties become significant. For example, when $\sum_u=90$, the proposed method improves total AI-RAN operator utility and social welfare over the best benchmark by about $56.8.9\%$ and $51.7\%$, respectively. Under representative user-type settings $N=8$ and $N=12$, the improvements are at least $89.6\%$ and $63.1\%$, respectively. Under representative refund settings, the proposed method improves total operator utility and social welfare by at least $3.7\%$ and $2.8\%$, respectively. 
% These results demonstrate that jointly considering contract design, user selection, and congestion feedback is effective for competitive AI-on-RAN service markets.

\bibliographystyle{IEEEtran}
\bibliography{zhan}

\end{document}